\documentclass[10pt]{article}
 
\usepackage[a4paper,left=18mm,top=25mm,right=18mm,bottom=25mm]{geometry}

\usepackage{authblk}
\usepackage[numbers, sort&compress]{natbib}
\usepackage[margin=30pt,font=small,labelfont=bf]{caption}
\usepackage{amsmath,amsopn,amsthm,amssymb}
\usepackage{thmtools,thm-restate}

\usepackage{graphicx}
\usepackage{subcaption}
\usepackage{mathptmx} 
\usepackage[mathscr]{euscript}
\usepackage{mathtools}
\usepackage{float}
\usepackage{enumitem}
\usepackage{xspace}
\usepackage{scalerel}
\usepackage{algorithm} 

\usepackage[
 	colorlinks=true,
	urlcolor=black,
	linkcolor=blue
]{hyperref}

\usepackage{algorithmicx}
\usepackage{algorithm} 
\usepackage[noend]{algpseudocode}
\algrenewcommand\algorithmicrequire{\textbf{Input:}}
\algrenewcommand\algorithmicensure{\textbf{Output:}}

\newtheorem{theorem}{Theorem}[section]
\newtheorem{lemma}[theorem]{Lemma} 
\newtheorem{corollary}[theorem]{Corollary}
\newtheorem{proposition}[theorem]{Proposition}
\newtheorem{example}{Example}
\newtheorem{definition}[theorem]{Definition}
\newtheorem{observation}[theorem]{Observation}

\newtheorem{remark}{Remark}

\newcommand{\tc}{\ensuremath{\operatorname{tc}}}
\newcommand{\lca}{\ensuremath{\operatorname{lca}}}
\newcommand{\LCA}{\ensuremath{\operatorname{LCA}}}

\newcommand{\AF}{\texttt{RF}}
\newcommand{\AFeq}{RF$^{\npreceq}$}
\newcommand{\AFlca}{RF$^{\lca}$}

\newcommand{\rel}{\trianglelefteq}
\newcommand{\altsrel}{S}
\newcommand{\axiom}[1]{\textnormal{\textbf{(#1)}}}
\newcommand{\pairs}{\mathcal{P}_2}
\newcommand{\support}{\ensuremath{\operatorname{supp}}}

\newcommand{\FCL}{\ensuremath{\operatorname{CL}_F}}
\newcommand{\Fcl}{\ensuremath{\operatorname{cl}_F}}
\newcommand{\cl}{\ensuremath{\operatorname{cl}}}
\newcommand{\FR}{F_{|R}}

\newcommand{\GG}{\mathscr{G}}
\newcommand{\Hasse}{\mathscr{H}}

\DeclareMathOperator{\indeg}{indeg}
\DeclareMathOperator{\outdeg}{outdeg}

\providecommand{\keywords}[1]{\textbf{\textit{Keywords: }} #1}

\title{Inferring Phylogenetic Networks from Required and Forbidden 
       LCA-Constraints}

\usepackage{authblk}

\author[1]{Patricia A.\ Ebert} 

\author[2]{Marc Hellmuth}

\affil[1]{Department of Mathematics, Faculty of Science,
  Stockholm University, SE-10691 Stockholm, Sweden} 

\affil[2]{
	Faculty of Mathematics and Computer Science, 	
  	Leipzig University, DE-04109 Leipzig, Germany} 

\date{\ }

\begin{document}
\sloppy

\maketitle

\abstract{ 
Least common ancestor (LCA) constraints encode relative-order information in directed acyclic
graphs (DAGs) and give rise to a natural constraint-realization problem. Phylogenetic networks
provide an important class of DAGs in which such constraints are used to represent local
information about evolutionary histories.
In this paper, we study the inference of DAGs and phylogenetic networks from LCA-constraints, which specify relative
positions of the LCAs associated with pairs of leaves.
While previous work has characterized when a set of required LCA-constraints can be
realized by a DAG or phylogenetic network,
it is natural to consider additional constraints that must be
explicitly avoided.
We therefore consider the realization problem for pairs $(R,F)$, where $R$ is a set of required
LCA-constraints and $F$ is a set of forbidden ones. Since there are several natural ways to
formalize what it means for a DAG to avoid a forbidden LCA-constraint, we study
three such variants. For each of them, we characterize exactly when there exists a
DAG or a phylogenetic network that realizes all constraints
in $R$ while avoiding all constraints in $F$ in the respective sense.
Our main characterization is based on a closure operator obtained from four elementary
inference rules.
Based on these characterizations, we derive polynomial-time algorithms that decide the existence
of such realizations and construct one whenever it exists.
All algorithms developed in this paper are implemented in the freely available Python package
\texttt{RealLCA}.
}

\smallskip
\noindent
\keywords{DAG, lowest common ancestor, forbidden relations, 
BUILD algorithm, closure operator, polynomial-time algorithm}

\section{Introduction}

Least common ancestors are a fundamental structural concept in rooted trees and directed acyclic
graphs (DAGs). Comparisons between their relative positions give rise to a natural realization
problem: given a collection of prescribed LCA-comparisons, determine whether there exists a DAG
satisfying them and, if so, construct such a DAG.

A rooted phylogenetic network is a directed acyclic graph with a single root and whose leaves represent the observed taxa (e.g., genes, species,  nucleotide sequences or languages)
while its internal vertices describe branching or reticulate ancestry, see
Figure~\hyperref[fig:example_introduction]{\ref*{fig:example_introduction}(a)} for an example. Such
networks extend phylogenetic trees by allowing vertices with more than one parent and can therefore
represent reticulate events such as hybridization or horizontal gene transfer. From an algorithmic perspective, a
central problem is to reconstruct such a network from ``local information'' about its structure.
Classical examples of such local information include small induced substructures such as triplets
\cite{JanssonEtAl2006,vanIersel2009,vanIersel2011,lev2-09,HuberEtAl2011,PoormohammadiEtAl2014} and
trinets~\cite{vanIersel2022,Semple2021,vanIersel2014}, relations derived from best
matches~\cite{SchallerEtAl2021b,Geiss2020,Geiss2019,KORCHMAROS2021397,KORCHMAROS2026308,KSS:25},
orthology and paralogy relations \cite{Lindeberg2025,Hellmuth2013,HNLLMS:15,
AltenhoffEtAl2019,Lafond2016,Lafond2014,LM:15,Lafond2013}, or constraints formulated in terms of
relative positions of least common ancestors (LCAs)~\cite{LAMSH:25,Aho:81}.  In all these settings, the underlying
algorithmic task is to reconstruct a global DAG or network from local structural constraints.

Here, we pursue the inference of DAGs and
phylogenetic networks from LCA-constraints. The underlying realization problem originates with Aho
et al.~\cite{Aho:81}, who studied, in a different context, whether a collection of prescribed
comparisons between least common ancestors can be represented by a rooted tree. Their
polynomial-time \texttt{BUILD} algorithm subsequently became a fundamental tool in phylogenetics, in
particular for deciding the compatibility of rooted triples and constructing a corresponding
phylogenetic tree; see, for example, \cite{SempleSteel:03}. More recently, Lindeberg et
al.~\cite{LAMSH:25} extended the original LCA-constraint framework from trees to DAGs and
phylogenetic networks.

To make this setting precise, let $N$ be a DAG or phylogenetic network with leaf set $X$. 
A least common ancestor (LCA) of two leaves $x,y\in X$ is a vertex $v$ of $N$ that is an ancestor of both $x$ and $y$ and such that no proper
descendant of $v$ has this property. We write $\lca_N(xy)$ whenever the LCA of $x$ and $y$ is
uniquely determined in $N$. LCAs are also frequently referred to as lowest common
ancestors \cite{Aho:81,vanIersel2011,JanssonEtAl2006}.

An LCA-constraint is an ordered pair $(ab,cd)$ of pairs of two leaves. Under
\emph{strict realization}, it requires $\lca_N(ab)$ to be a strict descendant of $\lca_N(cd)$ and
therefore expresses that $a$ and $b$ share a more recent common ancestor than $c$ and
$d$. The more
general notion of \emph{realization} introduced below also permits the two LCAs to coincide when the
input relation contains comparisons in both directions. Constraints of the simpler form $(ab,ac)$
are included as a special case and are closely related to the LCA-comparisons encoded by rooted
triples. Allowing arbitrary pairs $(ab,cd)$ is nevertheless essential, since constraints whose two
pairs share a leaf may imply further LCA-constraints involving disjoint pairs. In general,
LCA-constraints describe relative order relationships without requiring
phylogenetic parameter such as
divergence times, branch lengths, or a fully resolved network topology.

While Aho et al.~\cite{Aho:81} studied the realization of LCA-constraints by rooted trees,
Lindeberg et al.~\cite{LAMSH:25} generalized this framework to DAGs and phylogenetic networks by
characterizing strictly realizable and realizable relations and providing polynomial-time
recognition and construction algorithms. 

\begin{figure}[t]
    \centering
    \includegraphics[width=0.7\linewidth]{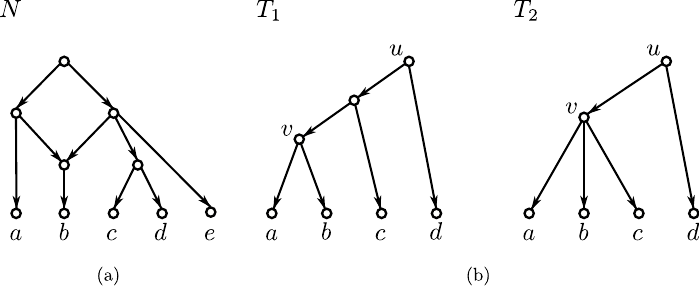}
    \caption{
    \textbf{(a)} A phylogenetic network $N$ on taxa set  $X = \{a,b,c,d,e\}$. 
    \textbf{(b)} Consider the two relations $R = \{(ab,cd)\}$ and $F = \{(ab,ac)\}$ representing required and forbidden LCA-constraints, respectively.
    In both phylogenetic trees $T_1$ and $T_2$, it holds that $\lca_{T_i}(ab)=v$ is a descendant of $\lca_{T_i}(cd)=u$ 
    and, in particular, that $T_1$ and $T_2$ realize the required LCA-constraint $(ab,cd)$. 
    For a DAG $G$ to satisfy the forbidden LCA-constraint $(ab,ac)$, it must hold that
    $\lca_G(ab)\nprec_G\lca_G(ac)$, provided that both LCAs are well-defined; cf.\ 
    Definition~\ref{def:AFreal}. Hence, $T_2$ satisfies the forbidden constraint $(ab,ac)$, since
    $\lca_{T_2}(ab)=v=\lca_{T_2}(ac)$. In contrast, $\lca_{T_1}(ab)\prec_{T_1}\lca_{T_1}(ac)$, and
    hence $T_1$ violates this forbidden LCA-constraint.
    }
    \label{fig:example_introduction}
\end{figure}

The required-only realization problem naturally extends to instances in which certain
LCA-comparisons must explicitly be avoided. Forbidden configurations have also been considered in
related realization problems, for instance for rooted triplets \cite{HHJS:06, Bryant1997}. This
motivates the study of a more expressive realization problem in which one is given not only a
relation $R$ of required LCA-constraints, but also a relation $F$ of forbidden LCA-constraints. The
task is to determine whether there exists a DAG or network that realizes all
constraints in $R$ while avoiding every constraint in $F$. A forbidden constraint $(ab,cd)\in F$
requires that $\lca_N(ab)$ is not a strict descendant of $\lca_N(cd)$ whenever both LCAs are
well-defined. Thus, equality, incomparability, and the reverse ancestral relation are allowed, see
Figure~\hyperref[fig:example_introduction]{\ref*{fig:example_introduction}(b)} for an example. In
Section~\ref{sec:alternative}, we also discuss two further natural interpretations of forbidden
constraints.

Our approach follows several ideas developed by Lindeberg et al.~\cite{LAMSH:25} for required
LCA-constraints. Nevertheless, the extension to forbidden constraints is not immediate. Required
constraints may imply additional comparisons, while avoiding a forbidden constraint may force two
LCA-classes to coincide. Capturing these interactions requires the new
$F$-conditional-symmetry rule and the associated closure $\Fcl(R)$.

In this contribution, we characterize the pairs $(R,F)$ for which there exists a 
DAG or network that realizes all constraints in $R$ while avoiding all constraints in $F$.
We study three natural interpretations of forbidden constraints and provide characterizations
for each of them. The characterization of all of these
notions is based on the closure operator $\Fcl$,
defined by four elementary inference rules. These characterizations yield polynomial-time
algorithms that decide whether such a realization exists and, in the affirmative
case, construct a realizing DAG and phylogenetic network.

This paper is organized as follows. In Section~\ref{sec:prelim}, we introduce the basic definitions
and notation used throughout the paper. In Section~\ref{sec:realization-required}, we recall
established results concerning the realization of sets $R$ of required LCA-constraints by DAGs and
phylogenetic networks. In Section~\ref{sec:AFLCA}, we extend this framework to pairs $(R,F)$ of
required and forbidden LCA-constraints. In particular, we introduce an additional inference rule
that captures the interaction between required and forbidden constraints, define the resulting
closure $\Fcl(R)$, and use it to characterize realizable pairs and construct corresponding
phylogenetic networks in polynomial time. In Section~\ref{sec:alternative}, we consider two further
natural interpretations of forbidden LCA-constraints and provide analogous characterizations under
these alternative notions of realizability. Finally, in Section~\ref{sec:classical-closure}, we
relate the rule-based closure $\Fcl(R)$ to the ``classical'' closure obtained by intersecting
the relations induced by all DAGs that realize $(R,F)$, and show that the two notions coincide
whenever $(R,F)$ is realizable.

All algorithms developed in this paper are implemented in the freely available Python package \texttt{RealLCA}, which is available at
\url{https://github.com/AnnaLindeberg/RealLCA}.

%
%
\section{Preliminaries}
\label{sec:prelim}

We begin with some basic definitions concerning sets and relations, illustrated by an example at the
end of this subsection, before turning to DAGs and networks.

\paragraph{Sets and Relations}

In what follows, $X$ will always be a finite non-empty set. We denote with $\mathcal{P}(X)$ the
power set of $X$. Moreover, we let $\pairs(X)\coloneqq\{\{a,b\} \mid a,b\in X\}$ (with $a=b$
allowed) denote the set system consisting of all 1- and 2-element subsets of $X$. We will often
write $ab$, respectively, $aa$ for elements $\{a,b\}$, respectively, $\{a\}$ in $\pairs(X)$. Thus,
$ab=ba$ always holds. 

Given a set $A$, a  subset $R\subseteq A\times A$ is a \emph{binary relation (on $A$)}.

\begin{remark}
As all relations considered in this work are binary, we shall simply refer to them as \emph{relations}.
\end{remark}

Furthermore, we define the \emph{support $\support_R$} of a   relation $R$ on $A$ as
\[
\support_R \coloneqq \{p\in A \mid \text{ there is some } q\in A \text{ with }  (p,q) \in R \text{ or } (q,p) \in R\},
\]   
that is, the subset of $A$ that contains precisely those $p\in A$ that are in $R$-relation with some
$q\in A$. We often consider relations $R$ on $A =\pairs(X)$ in which case we extend $\support_R$ to
obtain $\support_{R}^+ \coloneqq \support_{R}\ \cup\ \{xx\mid x\in X\}$.

Let $R$ be a relation on $A$. Then, $R$ is \emph{asymmetric} if $(p,q) \in R$ implies $(q,p) \notin
R$ for all $p,q\in A$, and it is \emph{anti-symmetric} if $(p,q) \in R$ and $(q,p) \in R$ implies
$p=q$ for all $p,q\in A$. Moreover, $R$ is \emph{transitive}, if $(p,q) \in R$ and $(q,r) \in R$
implies $(p,r) \in R$ for all $p,q,r\in A$. For a subset $B\subseteq A$, we say that $R$ is
\emph{$B$-reflexive} if $(b,b)\in R$ for every $b\in B$. A relation $R$ on $A$ is \emph{reflexive}
if it is $A$-reflexive. A \emph{poset} $(A, \leq)$ is a set $A$ equipped with a partial order
$\leq$, i.e., a relation $\leq$ on $A$ that is reflexive, transitive, and anti-symmetric.

A \emph{closure operator} on a set $S$ is a map $\phi \colon \mathcal{P}(S) \to\mathcal{P}(S) $ that
satisfies the following three properties for all $R,R'\in \mathcal{P}(S)$: \emph{Extensivity}: $R
\subseteq \phi(R)$, \emph{Monotonicity}: $R \subseteq R'$ implies $\phi(R) \subseteq \phi(R')$, and
\emph{Idempotency}: $\phi(\phi(R)) = \phi(R)$ \cite{CASPARD2003241,BS:95,SH:18}. We let $\tc(R)$
denote the \emph{transitive closure} of a relation $R$, that is, the inclusion-minimal relation that
is transitive and that contains $R$ (see e.g. \cite[p.39]{matouvsek2009invitation}). It is
straightforward to verify that $\tc$ is indeed a closure operator on $S = A\times A$ for all
relations $R$ on $A$. The following rules to modify a given binary relation will play a central role
in this paper. 

\begin{definition}\label{def:rules}
For two relations $S$ and $F$ on $\pairs(X)$, we define:
\begin{description}
  \item[\axiom{R1}] \emph{Reflexivity:} for all  $p\in \support^+_S$, add $(p,p)$ to $\altsrel$.
  \item[\axiom{R2}] \emph{Transitivity:} if $(p,q) \in \altsrel$ and $(q,r) \in \altsrel$, add $(p,r)$ to $\altsrel$.
  \item[\axiom{R3}] \emph{Cross-Consistency:} if $ab\in \support_S$ and $(ac,xy) \in \altsrel$ and $(bd,xy) \in \altsrel$ for some $c,d\in X$, add $(ab,xy)$ to $\altsrel$.
  \item[\axiom{R4}] \emph{$F$-Conditional-Symmetry:} if $(p,q) \in F$ and $(p,q) \in S$, add $(q,p)$ to $S$. 
\end{description}
\end{definition}

We say that a relation $S$ is \emph{cross-consistent} if \axiom{R3} applied on $S$ does not change
$S$. Similarly, $S$ is \emph{$F$-conditional-symmetric ($F$-csym)} w.r.t.\ some relation $F$ if
\axiom{R4} applied on $S$ does not change $S$. In other words, $S$ is $F$-csym if $(p,q)\in F \cap
S$ implies that $(q,p)\in S$. 

In our setting, an element $(p,q)$ of a relation is intended to represent the comparison
$\lca_G(p)\preceq_G \lca_G(q)$ in a DAG $G$, provided that the respective LCAs are well-defined. Hence, the rules
\axiom{R1} to \axiom{R4} should be viewed as inference rules: starting from a collection of known
LCA-comparisons, they add further comparisons that are necessarily implied.

Rule~\axiom{R1} reflects that every well-defined LCA is a descendant of itself. Rule~\axiom{R2} is
the usual transitivity rule. Thus, if $\lca_G(p)\preceq_G\lca_G(q) $ and
$\lca_G(q)\preceq_G\lca_G(r)$, then necessarily $ \lca_G(p)\preceq_G\lca_G(r)$. Rule~\axiom{R3}
captures an additional implication that is specific to LCAs. Suppose that $ab\in\support_S$ and
that, for suitable $c,d\in X$, $ \lca_G(ac)\preceq_G\lca_G(xy) $ and
$\lca_G(bd)\preceq_G\lca_G(xy)$. Since $\lca_G(xy)$ is then an ancestor of both $a$ and $b$, it is
also an ancestor of their least common ancestor. Consequently, $ \lca_G(ab)\preceq_G\lca_G(xy)$ holds,
whenever $\lca_G(ab)$ is well-defined; see also \cite[L~11]{LAMSH:25}. This explains the
cross-consistency rule. Finally, Rule~\axiom{R4} describes the interaction between required and
forbidden constraints. If $(p,q)$ is both implied by the current relation and forbidden by $F$, then
the comparison cannot be strict. The reverse comparison $(q,p)$ must therefore also hold, forcing
the two corresponding LCAs to coincide. This implication is formalized in
Lemma~\ref{lem:RF-lca-constraints}.
We conclude the preliminaries on sets and relations with an illustrative example.
\begin{example}
Let $X = \{a,b,c\}$ and consider the relation $R = \{(aa,ab), (ab,bc)\}$ on $\pairs(X) = \{aa,bb,cc,ab,ac,bc\}$. 
Its support is $\support_R = \{aa,ab,bc\}$ and, therefore, $\support^+_R = \{aa,bb,cc,ab,bc\}$. 
The relation $R$ is asymmetric and, consequently, anti-symmetric, but neither transitive nor reflexive.  
The transitive closure of $R$ is $\tc(R) = R \cup \{(aa,bc)\}$. 
Applying \axiom{R1} to $R$ yields $R' = R \cup \{(aa,aa),(bb,bb),(cc,cc),(ab,ab),(bc,bc)\}$. 
Hence, $R'$ is $\support^+_R$-reflexive, but  not reflexive, as $(ac,ac)\notin R'$.
Rule \axiom{R2} applied to $(aa,ab),(ab,bc) \in R$ yields $R \cup \{(aa,bc)\}$. 
This additional element $(aa,bc)$ can also be obtained by \axiom{R3}. To see this, simply replace 
the two pairs $(ac,xy)$ and $(bd,xy)$ in Definition~\ref{def:rules} by $(ab,bc)\in R$ and observe that  $aa \in \support_R$.
If we additionally consider a further relation $F = \{(ab,bc)\}$ on $\pairs(X)$, 
then \axiom{R4} and $(ab,bc) \in F \cap R$ yields $R \cup \{(bc,ab)\}$. 
\end{example}

\paragraph{DAGs and Networks}

A \emph{directed graph $G=(V,E)$} is a pair with non-empty vertex set $V(G)\coloneqq V$ and arc set
$E(G) \coloneqq E \subseteq V\times V$. We put $\outdeg_G(v)\coloneqq\left|\left\{u\in V \colon
(v,u)\in E\right\}\right|$ and $\indeg_G(v)\coloneqq\left|\left\{u\in V \colon (u,v)\in
E\right\}\right|$ to denote the \emph{out-degree} and \emph{in-degree} of a vertex $v$,
respectively. A vertex $v$ with $\outdeg_G(v)=0$ is a \emph{leaf} of $G$ and a vertex $v$ with
$\indeg_G(v)=0$ is a \emph{root} of $G$. A directed graph $G$ is \emph{phylogenetic} if it does not
contain a vertex $v$ such that $\outdeg_G(v)=1$ and $\indeg_G(v)\leq 1$. We sometimes use $u\to v$
to denote the arc $(u,v)\in E(G)$ and $u\leadsto v$ to denote a directed $uv$-path in $G$.

Directed graphs $G$ without directed cycles are called \emph{directed acyclic graphs (DAGs)}
\cite{book:digraph}. Let $G$ be a DAG. If $u\to v$ is an arc in $G$, then we call $v$ a \emph{child}
of $u$ and $u$ a \emph{parent} of $v$. We write $v\preceq_G u$ if and only if there is a directed
$uv$-path $u\leadsto v$ in $G$ and call, in this case, $v$ a \emph{descendant} of $u$ and $u$ an
\emph{ancestor} of $v$. If $v\preceq_G u$ and $v\neq u$, we write $v\prec_G u$. 

If $G$ is a DAG whose set of leaves is $X$, then $G$ is a \emph{DAG on $X$}. A \emph{network} is a
DAG with a single root. A \emph{(rooted) tree} is a network that does not contain vertices with
$\indeg_G(v)>1$.

For a given DAG $G$ on $X$ and a non-empty subset $A\subseteq X$, a vertex $v\in V(G)$ is a
\emph{common ancestor of $A$} if $v$ is an ancestor of every vertex in $A$. Moreover, $v$ is a
\emph{least common ancestor} (LCA) of $A$ if $v$ is a $\preceq_G$-minimal vertex that is an ancestor
of all vertices in $A$. The set $\LCA_G(A)$ comprises all LCAs of $A$ in $G$. In a network $N$ on
$X$, the unique root is a common ancestor for all $A\subseteq X$ and, therefore,
$\LCA_N(A)\neq\emptyset$. Moreover, we are interested in DAGs where $|\LCA_G(A)|=1$ holds for
certain subsets $A\subseteq X$. For simplicity, we will write $\lca_G(A)=v$ in case that
$\LCA_G(A)=\{v\}$ and say that \emph{$\lca_G(A)$ is well-defined}. 
To recall, we often write $xy$ for sets $\{x,y\}$, which allows us to put
$\lca_G(xy) \coloneqq \lca_G(\{x,y\})$, if $\lca_G(\{x,y\})$ is well-defined. A DAG $G$ on $X$ is
\emph{2-lca-relevant} if, for all $v\in V(G)$, there are (not necessarily distinct) leaves $x,y\in
X$ such that $v=\lca_G(xy)$.

For a given DAG $G$ on $X$, we further define the relation $\rel_G$ on $\pairs(X)$ as \[\rel_G
\coloneqq \{(ab,xy) \mid \lca_G(ab), \lca_G(xy) \text{ are well-defined and } \lca_G(ab)\preceq_G
\lca_G(xy)\}.\] 
Several of the latter definitions are exemplified in Figure~\ref{fig:prelim_dags}.

\begin{figure}[t]
    \centering
    \includegraphics[width=0.7\linewidth]{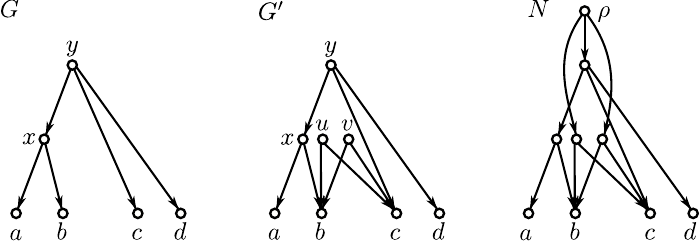}
    \caption{Shown are three phylogenetic DAGs $G$, $G'$, and $N$ on $X=\{a,b,c,d\}$.
    Since $G'$ has more than one root, it is not a network. In contrast, $G$ and $N$ have the unique roots $y$ and $\rho$, respectively, and are therefore networks.
    Here $N$ is obtained from $G'$ by adding a  root $\rho$ and arcs $\rho \to w$ for the roots $w$ of $G'$
    as specified in Lemma~\ref{lem:DAG2Network}. Moreover, $G'$ is the $bc$-extension of $G$.
    In $G$, the vertices $x$ and $y$ are the only common ancestors of $a$ and $b$. 
    Since, in addition $x \prec_G y$, the vertex $x$ is the unique least common ancestor. 
    Consequently, $\lca_G(ab) = x$ is well-defined. 
    In $G'$, the LCA $\lca_{G'}(bc)$ is not well-defined, as $\LCA_{G'}(bc) = \{u,v,y\}$. 
    The network $G$ is 2-lca-relevant while $G'$ and $N$ are not.} 
    \label{fig:prelim_dags}
\end{figure}

For a poset $({Q},\leq)$, the \emph{Hasse diagram} $\Hasse(Q,\leq)$ is the DAG with
vertex set $Q$ and arcs $(A,B)$ if (i) $B\leq A$ and $A\neq B$ and (ii) there is no $C\in Q$ with
$B\leq C\leq A$ and $C\neq A,B$. Figure~\ref{fig:hasse_diagram} shows the Hasse diagram of some poset. 

\begin{figure}[t]
    \centering
    \includegraphics[width=0.7\linewidth]{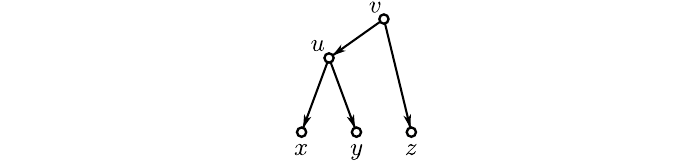}
    \caption{Let $Q = \{u,v,x,y,z\}$ and $\leq = \{(u,u),(v,v),(x,x), (y,y), (z,z),(x,u),$ $(x,v),(y,u),(y,v),(u,v),(z,v)\}$.
     Since $\leq$ is a reflexive, transitive, and anti-symmetric relation on $Q$, the pair $(Q,\leq)$ is a poset. Shown is the Hasse diagram $\Hasse(Q,\leq)$. }
    \label{fig:hasse_diagram}
\end{figure}

Throughout this paper, we often transform a DAG into a network and will use the following simple
result.

\begin{lemma}[From DAGs to networks]\label{lem:DAG2Network}
Let $G$ be a DAG on $X$. Let $N$ be the directed graph obtained from $G$ by either putting
$N\coloneqq G$ if $G$ is a network or by adding a new vertex $\rho$ to $G$ together with the arcs
$(\rho,\rho_i)$ for all roots $\rho_1,\dots,\rho_k$, $k\geq 2$ of $G$. Then, $N$ is a network on $X$
such that $u\preceq_G v$ if and only if $u\preceq_N v$ for all $u,v\in V(G)$. In particular, if
$\LCA_G(xy)\neq \emptyset$, then $\LCA_G(xy)=\LCA_N(xy)$. Hence, if $\lca_G(xy)$ is well-defined in
$G$, then $\lca_N(xy)=\lca_G(xy)$. Moreover, if $G$ is phylogenetic, then $N$ is phylogenetic.
\end{lemma}
\begin{proof}
	Let $G$ be a DAG on $X$ and let $N$ be as described in the statement. If $N = G$, then the
	statement clearly holds. Thus, assume that $G$ is not a network. Since no leaves and only outgoing
	arcs of the new vertex $\rho$ are added, $N$ is a DAG with the unique root $\rho$. Hence, $N$ is a
	network on $X$. Note that $V(G)=V(N)\setminus\{\rho\}$. This and the construction of $N$ implies
	that $v\prec_N\rho$ holds for all $v\in V(G)$ as well as $u\preceq_G v$ if and only if $u\preceq_N
	v$ for all $u,v\in V(G)$. Now let $x,y \in X$ with $\LCA_G(xy) \neq \emptyset$. Note that
	$\rho\notin \LCA_G(xy)\subseteq V(G)$. By the latter arguments, $w\prec_N\rho$ for all $w\in
	\LCA_G(xy) \subseteq V(G)$. This and $u\preceq_G v$ if and only if $u\preceq_N v$ for all $v\in
	V(G)$ implies that $\LCA_G(xy) = \LCA_N(xy)$. In particular, if $\lca_G(xy)$ is well-defined in
	$G$, then $\lca_N(xy)=\lca_G(xy)$. Suppose $G$ is phylogenetic, i.e., there exists no $u \in V(G)$
	such that $\outdeg_G(u) = 1$ and $\indeg_G(u) \leq 1$. By construction, $\outdeg_G(v) =
	\outdeg_N(v)$ holds for all $v \in V(G)$ and $\indeg_G(v) = \indeg_N(v)$ for all $v \in V(G)
	\setminus \{\rho_1, \dots, \rho_k\}$, where $\rho_1, \dots, \rho_k$ are the roots of $G$.
	Moreover, $\indeg_G(v) \leq \indeg_N(v)$ holds for all $v \in \{\rho_1, \dots, \rho_k\}$. Lastly,
	since $G$ is not a network, $k\geq 2$ and, thus, $\outdeg_N(\rho) \geq 2$. Hence, $N$ is
	phylogenetic. 
\end{proof}

The construction of a network $N$ as specified in Lemma~\ref{lem:DAG2Network} is illustrated
in Figure~\ref{fig:prelim_dags}.
We shall use the following local construction to avoid constraints involving a pair $xy$
whose LCA is not required to be unique. For two distinct leaves $x,y$ of a DAG $G$, the
\emph{$xy$-extension}  of $G$ is the digraph
obtained from $G$ by adding  two new vertices $u,v$ 
and the arcs $\{(u,x),(u,y),(v,x),(v,y)\}$ to $G$; see Figure~\ref{fig:prelim_dags} for an example.
Let $G'$ be the $xy$-extension of the DAG $G$.
One easily verifies that no cycles are introduced in $G'$ by the additional vertices and arcs 
and it follows that $G'$ remains a DAG.
Moreover, by construction, $u$ is a parent of $x$ and $y$ and $u$ does not
have any further children. Hence, $u$ (and by similar arguments, $v$) is an LCA of $x$ and $y$ in $G'$.
Additionally, it is a straightforward task to verify that in $G'$, it holds by construction that $a \preceq_G b$ 
if and only if $a \preceq_{G'} b$ for all $a,b \in V(G) =V(G')\setminus \{u,v\}$. Moreover, $u$ and $v$ have only $x,
y$, and $u$ resp. $v$ as descendants and no ancestors. The latter arguments imply $\LCA_G(ab) =
\LCA_{G'}(ab)$ for all $a,b \in X$ with $\{a,b\} \neq \{x,y\}$. We summarize this discussion into
\begin{observation}\label{obs:xy-extension1}
Let $G$ be a DAG on $X$ and $x,y \in X$ be distinct leaves.
The $xy$-extension $G'$ of $G$ is a DAG on $X$ and it holds that 
\begin{itemize}
    \item[(i)] $|\LCA_{G'}(xy)|>1$ and, thus, $\lca_{G'}(xy)$ is not well-defined, and
    \item[(ii)] $a \preceq_G b$ if and only if $a \preceq_{G'} b$ for $a,b \in V(G)$, and
    \item[(iii)] $\LCA_G(ab) = \LCA_{G'}(ab)$ for all $a,b \in X$ with $\{a,b\} \neq \{x,y\}$. 
\end{itemize}
\end{observation}

%
%

\section{Realization of Required  LCA-constraints}
\label{sec:realization-required}

We begin by recalling the main definitions and results concerning realizable relations introduced
by Lindeberg et al.~\cite{LAMSH:25}. These results form the basis for our extension to forbidden
LCA-constraints in Section~\ref{sec:AFLCA}.

Our starting point is a classical result due to Aho et al.~\cite{Aho:81}. They considered the
following problem: given a collection of constraints $(ab,cd)\in R$, where $a\neq b$, $c\neq d$,
and $a,b,c,d\in X$, can one construct a phylogenetic tree $T$ on $X$ such that
\[
\lca_T(ab)\prec_T\lca_T(cd)
\]
for every $(ab,cd)\in R$, or determine that no such tree exists? Thus, each constraint prescribes
that the taxa $a$ and $b$ have a more recent common ancestor than the taxa $c$ and $d$. This leads
to the following notion.

\begin{definition}[Strict Realization,  {\cite[Def~3]{LAMSH:25}}]\label{def:strict-real}
	A relation $R$ on $\pairs(X)$ is \emph{strictly realizable} if there is a DAG $G$ on $X$ such that
	for all $ab,cd\in\support_R^+$, the vertices $\lca_G(ab)$ and $\lca_G(cd)$ are well-defined and
	the following implication holds:
	 \begin{description}
      \item[\axiom{I0}] $(ab,cd) \in R$ implies that $\lca_G(ab)\prec_G\lca_G(cd)$.
   \end{description}
   \noindent In this case, we say that $R$ \emph{is strictly realized by} $G$. 
\end{definition}

Strict realization does not allow two pairs occurring in a constraint to have the same LCA. 
A natural first relaxation would be to require only that
\[
(ab,cd)\in R
 \implies 
\lca_G(ab)\preceq_G\lca_G(cd).
\]
This condition alone is, however, too permissive: whenever $R$ contains no pair $(ab,xx)$ with
$ab\neq xx$, the star tree would satisfy all constraints by assigning the same LCA to every
non-singleton pair.

The following definition avoids this problem by using the transitive closure of $R$ to distinguish
between genuinely strict comparisons and pairs of constraints that force equality. It therefore
extends strict realization while retaining meaningful information about the relative positions of
the LCAs.

\begin{definition}[Realization, {\cite[Def~4]{LAMSH:25}}]\label{def:real}
A relation $R$ on $\pairs(X)$ is \emph{realizable} if there is a DAG $G$ on $X$ such that for all 
$ab, cd \in \support^+_{R}$, the vertices $\lca_G(ab)$ and $\lca_G(cd)$ are well-defined and the 
following implications hold: 
\begin{description}
    \item[\axiom{I1}] $(ab,cd)\in R$ and $(cd,ab) \notin \tc(R)$ implies that $\lca_G(ab) \prec_G \lca_G(cd)$. 
    \item[\axiom{I2}] $(ab,cd)\in R$ and $(cd,ab) \in \tc(R)$ implies that $\lca_G(ab) = \lca_G(cd)$.  
\end{description}
In this case, we say that $R$ \emph{is realized by} $G$. 
\end{definition}

Condition~\axiom{I1} applies whenever the comparison represented by $(ab,cd)$ is not ``reversible'' in
$\tc(R)$ and therefore requires a strict ancestral relation. In contrast, if both pairs 
 $(ab,cd)$ and $(cd,ab)$  are
present in $\tc(R)$, then Condition~\axiom{I2} identifies the corresponding LCAs
$\lca_G(ab)$ and $\lca_G(cd)$. Thus, all elements of any subset $\mathcal{P}\subseteq \pairs(X)$
such that both $(p,q)$ and $(q,p)$ belong to $\tc(R)$ for all $p,q\in\mathcal{P}$
 must share the same LCA in every realization.
Figure~\ref{fig:realization} shows a DAG $G$ that realizes a relation $R$ and also illustrates
the remaining statements of this section. The notion of strict realizability introduced by
Aho et al.\ is just a special case of realizability, as shown by the next result.

\begin{lemma}[{\cite[L~5]{LAMSH:25}}]\label{lem:3.3-lamsh}
   A relation $R$ is strictly realized by $G$ if and only if $R$ is realized by $G$ and $\tc(R)$ is asymmetric. 
\end{lemma}

\begin{figure}[h]
    \centering
    \includegraphics[width=0.7\linewidth]{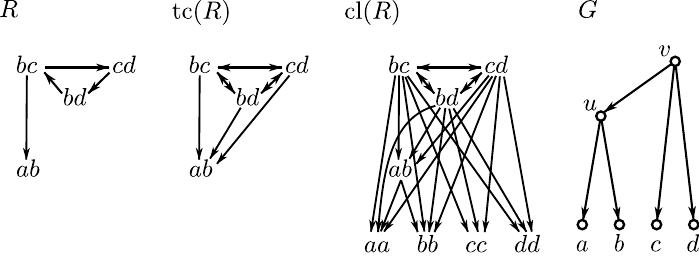}
    \caption{
    Let $R = \{(ab,bc),(bc,bd),(bd,cd),(cd,bc)\}$ be a relation on 
    $\pairs(X)$ with $X = \{a,b,c,d\}$. 
    The graphical representations of $R$, $\tc(R)$, and $\cl(R)$ are shown, with an arc $p\to q$ drawn when $(q,p)$ belongs to the respective relation.
    Note that we omitted the reflexive arcs $(p,p)$ in the drawing of $\tc(R)$ and $\cl(R)$.
    It is straightforward to verify that $R$ satisfies \axiom{X1} and \axiom{X2}. 
    By Theorem~\ref{thm:char}, $R$ is thus realizable. 
    In fact, the DAG $G$ realizes $R$, since for $(ab,bc) \in R$ and $(bc,ab) \notin \tc(R)$, it holds that $\lca_G(ab) \prec_G \lca_G(bc)$ according to \axiom{I1}, 
    and for $(bc,bd),(bd,cd),(cd,bc) \in R$ and $(bd,bc),(cd,bd),(bc,cd) \in \tc(R)$, we have $\lca_G(bc) = \lca_G(bd) = \lca_G(cd)$ according to \axiom{I2}.
    Since $\tc(R)$ is not asymmetric, Lemma~\ref{lem:3.3-lamsh} implies that $R$ is not strictly realizable. 
    }
    \label{fig:realization}
\end{figure}

The previous definitions describe what a realizing DAG must satisfy, but they do not yet provide a
direct way to decide whether such a DAG exists. 
Given a relation $R$ on $\pairs(X)$, 
Lindeberg et al. \cite{LAMSH:25} showed that
central to the recognition problem is the
closure  
\[
\cl(R) \coloneqq \bigcap_{R' \in \mathfrak{R}_{R}} R',
\]
with $\mathfrak{R}_{R}$ being 
the set of all relations $R'$ on $\pairs(X)$ that are $\support^+_{R}$-reflexive, transitive,
cross-consistent, and satisfy $R\subseteq R'$.
If $R$ is realizable, then \cite[Thm~44]{LAMSH:25} shows that 
\[
\cl(R) = \bigcap_{G\in \mathbb{G}} \rel_G, 
\]
where $\mathbb{G}$ denotes the set of all DAGs $G$ on
$X$ that realize $R$.
Informally, $\cl(R)$ contains all comparisons that are implied by constraints in $R$ and, in case $R$ is realizable, that are satisfied by every realization.
The map $\cl$ is a closure operator and therefore satisfies extensivity, monotonicity, and
idempotency. Importantly, its definition as an intersection does not require enumerating all
realizations. Instead, $\cl(R)$ can be computed in polynomial time by repeatedly applying the
three inference rules \axiom{R1}, \axiom{R2}, and \axiom{R3} from
Definition~\ref{def:rules}. Reflexivity and transitivity supply the usual consequences of a binary
relation, while cross-consistency captures additional implications imposed by the LCA structure.
Moreover, realizability of a relation $R$ can be
characterized in terms of two conditions, \axiom{X1} and \axiom{X2}, which also yield a
polynomial-time recognition and construction procedure.

\begin{theorem}[{\cite[Thm~17, 31 \& 38]{LAMSH:25}]}]\label{thm:char}
For any relation $R$ on $\pairs(X)$, $\cl(R)$ can be computed in polynomial time in $|X|$
as follows: Put $\altsrel = R$ and first apply the rule $\axiom{R1}$ to $\altsrel$. Afterwards,
repeatedly apply one of the two rules \axiom{R2} and \axiom{R3} to $\altsrel$ in any order, until
none of the rules can be applied. Then, $\altsrel=\cl(R)$. 

\smallskip
\noindent
Furthermore, the following statements are equivalent for any  relation $R$ on $\pairs(X)$:
\begin{enumerate}
  \item[(1)] $R$ is realizable.
  \item[(2)] $R$ satisfies the following two conditions: 
            \begin{description}
            \item[\axiom{X1}]  For all $a,b,x\in X$: $ab\neq xx$ implies 
            $(ab, xx)\notin R$. 
            \item[\axiom{X2}]  For all $a,b,x,y\in X$: 
            $(ab, xy) \in R$ 
            and $(xy, ab)\not\in \tc(R)$ implies $(xy,ab)\notin \cl(R)$. 
            \end{description}
\end{enumerate}
Moreover, in polynomial time in $|X|$, it can be verified if $R$ is realizable and, in the 
affirmative case, a DAG and network realizing $R$ can be constructed. 
\end{theorem}

Thus, the closure  yields a characterization of realizability. Condition~\axiom{X1} excludes a
non-singleton pair from lying below a singleton pair, which is impossible because a singleton pair $aa$ corresponds to the leaf $a$ 
in each realization, which can only be
an ancestor of itself. Condition~\axiom{X2} ensures that a comparison required to be strict is not
``reversed'' by any consequence generated by the inference rules defining the closure.

Theorem~\ref{thm:char} provides the required-only counterpart of the results developed in the next
section. There, the additional relation $F$ of forbidden constraints leads to a fourth inference
rule and to the modified closure $\Fcl(R)$.

%
%

\section{Realization of Required and Forbidden LCA-constraints}
\label{sec:AFLCA}

We now consider pairs $(R,F)$ of relations, where $R$ specifies required LCA-constraints and $F$ specifies forbidden ones, 
and introduce a corresponding notion of realizability, called \AF-realizability.
To characterize \AF-realizability, we develop two complementary approaches. The first shows that it
suffices to consider all required LCA-constraints together with a suitable subset of the forbidden
constraints. For the second approach, we construct a set $\Fcl(R)$ containing all LCA-constraints
implied by $R$ and $F$ and prove that $\Fcl$ defines a closure operator. Based on this closure, we
formulate two conditions, \axiom{Y1} and \axiom{Y2}, and show that they characterize
\AF-realizability.
Together, these approaches yield a polynomial-time algorithm that recognizes \AF-realizable pairs of
relations and, whenever such a realization exists, constructs an \AF-realizing DAG. In the next
subsection, we begin with the formal definition of \AF-realization before presenting the first
characterization and introducing the closure operator $\Fcl$.

\subsection{\AF-Realizability and the Closure}

In what follows, let $(R,F)$ be a pair of relations on $\pairs(X)$, that is, let both $R$ and $F$ be
relations on $\pairs(X)$. In our context, $R$ represents ``required (\texttt{R})'' LCA-constraints,
whereas $F$ represents ``forbidden (\texttt{F})'' LCA-constraints. This motivates the following
definition of realization.

\begin{definition}[\AF-Realization]\label{def:AFreal}
Let $(R, F)$ be a pair of relations on $\pairs(X)$. Then, $(R, F)$ is \emph{\AF-realizable} if there
is a DAG $G$ on $X$ such that $R$ is realized by $G$ and the following implication holds: 
\begin{itemize}
    \item[\axiom{F}] If $(ab,cd) \in F$ and $\lca_G(ab)$ and $\lca_G(cd)$ are well-defined, then  $\lca_G(ab) \nprec_G \lca_G(cd)$. 
\end{itemize}
In this case, we say that $(R,F)$ \emph{is \AF-realized by} $G$.  
\end{definition}

Thus, a forbidden constraint does not prescribe a specific alternative relation between the two
LCAs. It only rules out one strict ancestral comparison. In particular, the constraint is satisfied
if the two LCAs coincide, are incomparable, occur in the reverse order, or if at least one of them is
not well-defined.
An example of a DAG $\AF$-realizing a pair $(R,F)$ is provided in Figure~\ref{fig:working-ex1}.
In Section~\ref{sec:alternative}, we discuss two alternative definitions, one to forbid $\lca_G(ab)
\preceq_G \lca_G(cd)$ rather than $\lca_G(ab) \prec_G \lca_G(cd)$ for elements $(ab,cd) \in F$ and
one to ensure that $\lca_G(ab)$ is well-defined for all $ab \in \support_F$.

To characterize \AF-realizable pairs of relations,  we first show that it suffices to consider all
required LCA-constraints together with only a subset of the forbidden ones. To this end, observe
that, in a DAG $G$ that \AF-realizes a pair $(R,F)$ of relations, the LCA $\lca_G(ab)$  is not always required to be well-defined.
In this case, it must hold that $ab \notin \support^+_R$. In
particular, if $ab \notin \support^+_R$ but $ab \in \support_F$, we can exploit this by constructing
a DAG in which $\lca_G(ab)$ is not well-defined; then, none of the forbidden LCA-constraints
$(ab,cd) \in F$ or $(cd,ab) \in F$ are realized, as they trivially satisfy \axiom{F}. On the other
hand, if $ab \in \support^+_R$, then $\lca_G(ab)$ must be well-defined in any DAG $G$ realizing $R$,
regardless of whether $ab \in \support_F$ holds or not. Hence, to construct a DAG $G$ that
\AF-realizes $(R,F)$, it is sufficient to ensure that $\lca_G(ab)$ and $\lca_G(cd)$ are well-defined
only for those pairs $(ab,cd) \in F$ with $ab, cd \in \support^+_R$. To formalize the latter idea,
let $(R,F)$ be a pair of relations on $\pairs(X)$.
We restrict $F$ to those forbidden
constraints whose two LCA-pairs already occur in the required part and must therefore have
well-defined LCAs in every realization of $R$. We define the following relation on
$\pairs(X)$:
\[
\FR \coloneqq \{(ab,cd) \in F \mid ab,cd \in \support^+_R\}. 
\]
Moreover, the \emph{$\FR$-extension of a DAG $G$ on $X$} is obtained from $G$ by stepwise application of
$xy$-extensions for all $xy \in \support_F \setminus \support^+_R$. Applying an $xy$-extension to
$G$ yields a DAG $G'$ in which, by Observation~\ref{obs:xy-extension1}, $\lca_{G'}(xy)$ is not
well-defined. In this case, Condition \axiom{F} is satisfied for $G'$ and all elements $(ab,xy)$ and
$(xy,ab)$ in $F$. We emphasize that, if $(R,F)$ is a pair of relations on $\pairs(X)$, then $xx\in
\support_R^+$ for all $x\in X$. Hence, we obtain

\begin{observation}\label{obs:FR-extension}
    If $(R,F)$ is a pair of relations on $\pairs(X)$, 
    then in any $\FR$-extension of a DAG $G$ on $X$ and, thus, in the application of $xy$-extensions 
    it holds that $x\neq y$ for all $xy \in \support_F \setminus \support^+_R$.
\end{observation}

We are now in the position to provide a first characterization of \AF-realizable pairs $(R,F)$ in
terms of $R$ and the restriction $\FR$ of $F$. 

\begin{proposition}\label{prop:RF_real_iff_RFR_real}
Let $(R,F)$ be a pair of relations on $\pairs(X)$. 
Then, the following statements hold. 
\begin{enumerate}
    \item[(1)] If $(R,F)$ is \AF-realized by a DAG $G$, then $(R,\FR)$ is also \AF-realized by $G$.
    \item[(2)] If $(R,\FR)$ is \AF-realized by a DAG $G$, then $(R,F)$ is \AF-realized by the $\FR$-extension of $G$. 
\end{enumerate}
In particular, $(R,F)$ is \AF-realizable if and only if $(R,\FR)$ is \AF-realizable. 
\end{proposition}
\begin{proof}
	Let $(R,F)$ be a pair of relations on $\pairs(X)$. Clearly, if $(R,F)$ is \AF-realized by a DAG
	$G$, then $\FR \subseteq F$ implies that $G$ also \AF-realizes $(R,\FR)$. This, in particular,
	proves Statement~(1). 

	We show now that Statement~(2) holds. Hence, suppose that $(R,\FR)$ is \AF-realized by the DAG
	$G$. Let $(ab,cd) \in F \setminus \FR$. By the definition of $\FR$, we have $ab\notin
	\support_R^+$ or $cd\notin \support_R^+$. In other words, none of the elements in $R$ enforces
	both $\lca_H(ab)$ and $\lca_H(cd)$ to be well-defined in any DAG $H$ that \AF-realizes $(R,\FR)$.
	Let $pq\in \{ab,cd\}$ be such that $pq\notin \support_R^+$ and let $G'$ be the graph obtained from
	$G$ by applying a $pq$-extension. Observation~\ref{obs:xy-extension1} implies that $G'$ is a DAG
	and, moreover, that $\lca_{G'}(pq)$ is not well-defined. In addition,
	Observation~\ref{obs:xy-extension1}(ii) and (iii) and the definition of $\FR$ implies that $G'$
	\AF-realizes $(R,\FR)$. Together with $\lca_{G'}(pq)$ not being well-defined and $pq\in
	\{ab,cd\}$, it follows that $G'$ \AF-realizes $(R,\FR \cup \{(ab,cd)\})$. Now let $p'q'\in
	\{ab,cd\}\setminus \{pq\}$ and assume that $p'q' \notin \support^+_R$. In this case, we would not
	only apply a $pq$-extension but also a $p'q'$-extension. However, since $G'$ is a DAG that
	\AF-realizes $(R,\FR \cup \{(ab,cd)\})$, we can re-use an analogous argument as used for $G$ and
	$G'$ to conclude that the $p'q'$-extension $G''$ of $G'$ \AF-realizes $(R,\FR \cup \{(ab,cd)\})$.
	Repeating this procedure for all $(ab,cd) \in F \setminus \FR$, results in a DAG $H$ that
	\AF-realizes $(R, F)$. Note that $H$ corresponds to the $\FR$-extension of $G$ and, therefore,
	Statement~(2) holds.

	Finally Statement~(1) and (2) together imply that $(R,F)$ is \AF-realizable if and only if
	$(R,\FR)$ is \AF-realizable. 
\end{proof}

While interesting in its own right, the characterization in
Proposition~\ref{prop:RF_real_iff_RFR_real} will also be useful later in
Section~\ref{subs:charAFreal}, where we establish a further characterization
of \AF-realizability in terms of two simple conditions, \axiom{Y1} and
\axiom{Y2}, which are similar in flavor to \axiom{X1} and \axiom{X2}.
For now, we continue by deriving further properties of \AF-realizable pairs
$(R,F)$. In particular, we construct a set $\Fcl(R)$ containing all 
LCA-constraints implied by $R$ and $F$, and show that the resulting operator is a
closure operator. Using this closure, we then prove that \axiom{Y1} and
\axiom{Y2} characterize \AF-realizability.

To construct $\Fcl(R)$, we first examine the implications of LCA-constraints that are simultaneously
required and forbidden. At first sight, one might expect that required and forbidden relations must be
disjoint. This, however, need not be the case: the same ordered pair may occur in
both relations whenever the corresponding LCAs are forced to coincide. The following result makes
this observation precise.

\begin{lemma}\label{lem:RF-lca-constraints}
Let $(R,F)$ be a pair of relations on $\pairs(X)$. If $(R, F)$ is \AF-realized by $G$ and
$(ab,xy)\in R\cap F$, then $\lca_G(xy) = \lca_G(ab)$ and it holds that $(xy,ab)\in \tc(R)$.
\end{lemma}
\begin{proof}
	Suppose that $(R, F)$ is \AF-realized by $G$ and that $(ab,xy)\in R\cap F$. Assume, for
	contradiction, that $(xy,ab)\notin \tc(R)$. Since $R$ is realized by $G$, the vertices
	$\lca_G(ab)$ and $\lca_G(xy)$ are well-defined and \axiom{I1} implies that $\lca_G(ab) \prec_G
	\lca_G(xy)$. However, since $G$ \AF-realizes $(R, F)$ and $(ab,xy)\in F$, Condition \axiom{F} must
	hold and, therefore, $\lca_G(ab) \nprec_G \lca_G(xy)$; a contradiction. Hence, $(xy,ab)\in
	\tc(R)$. Thus, \axiom{I2} must hold and we can conclude that $\lca_G(xy) = \lca_G(ab)$. 
\end{proof}

In particular, Lemma~\ref{lem:RF-lca-constraints} motivates the definition of
$F$-conditional-symmetry as specified in Section~\ref{sec:prelim}. To recall, for a pair $(R,F)$ of
relations on $\pairs(X)$, $R$ is $F$-conditional-symmetric (\emph{$F$-csym}, for short) if for all
$a,b,x,y \in X$ (not necessarily distinct) the following statement holds: 
\[
(ab,xy)\in  R \cap F \text{ implies } (xy,ab)\in R. 
\]   

Application of rule \axiom{R4} enforces $(xy,ab) \in R$ whenever $(ab,xy) \in R \cap F$ and ensures,
according to Lemma~\ref{lem:RF-lca-constraints}, that $\lca_G(xy) = \lca_G(ab)$ for any DAG $G$ that
\AF-realizes $(R,F)$. The latter guarantees that the 
strict comparison forbidden by $(ab,xy)$ does not hold in $G$.

As proven in the following results, repeated application of the rules \axiom{R1}--\axiom{R4} to
$(R,F)$ yields the closure $\Fcl(R)$ of $R$ w.r.t.\ $F$, analogous to the closure $\cl(R)$ defined
in Section~\ref{sec:realization-required}. The first three rules propagate consequences of the
required relation, while the new Rule~\axiom{R4} records equalities that become necessary when a
comparison is both implied and forbidden. Intuitively, $\Fcl(R)$ therefore comprises all 
LCA-constraints implied by the constraints contained in $R$ and $F$. As we shall see, this closure plays
a crucial role both in characterizing \AF-realizable pairs and in constructing, whenever possible, a
DAG that \AF-realizes $(R,F)$.

\begin{definition}[The set $\mathfrak{R}_{R,F}$ and the relation $\Fcl(R)$] \label{Fclosure}
Let $(R,F)$ be a pair of relations on $\pairs(X)$. We define $\mathfrak{R}_{R,F}$ as the set of all
relations $R'$ on $\pairs(X)$ that are $\support^+_{R}$-reflexive, transitive, cross-consistent,
$F$-csym, and satisfy $R\subseteq R'$. Moreover, we define
\[
\Fcl(R) \coloneqq \bigcap_{R' \in \mathfrak{R}_{R,F}} R'.
\]
\end{definition}
Note that if $R$ and $F$ are relations on $\pairs(X)$, then $\Fcl(R)$ is considered as a relation on
$\pairs(X)$. Furthermore, observe that if $F=\emptyset$, the latter intersection is taken among
those relations that contain $R$ and are $\support^+_{R}$-reflexive, transitive, and
cross-consistent and, thus, $\cl_{\emptyset}(R)=\cl(R)$. In particular, for all relations $F$ it
holds that $\mathfrak{R}_{R,F} \subseteq \mathfrak{R}_{R,\emptyset}$. We summarize this discussion
into 
\begin{observation}\label{obs:cl_emptyset=cl}
For all relations $R$ and $F$ on $\pairs(X)$, it holds that $\cl(R) = \cl_{\emptyset}(R) \subseteq
\Fcl(R)$. 
\end{observation}

As we shall see in Proposition~\ref{prop:closure_operator}, $\Fcl$ is indeed a closure operator. To
establish this result, we first show that $\Fcl(R)$ can be computed in polynomial time using the
four simple rules \axiom{R1}--\axiom{R4} specified in Definition~\ref{def:rules}, 
thereby generalizing \cite[Thm.~16]{LAMSH:25}. More precisely, the next result provides a purely combinatorial procedure for
computing $\Fcl(R)$: starting with $R$, we repeatedly add all consequences forced by the four
inference rules. Since no pair is ever removed, this process terminates as soon as none of the rules
produces a new element.

\begin{theorem} \label{thm:polynomial_construction_Fclosure}
Let $(R,F)$ be a pair of relations on $\pairs(X)$. Let $S$ be a relation obtained from $(R,F)$ by
starting with $S = R$ and first applying the rule $\axiom{R1}$ to $\altsrel$. Afterwards, repeatedly
apply one of the three rules \axiom{R2}, \axiom{R3}, and \axiom{R4} to $(\altsrel,F)$ in any order,
until none of the rules can be applied. Then, $S = \Fcl(R)$ and $\support^+_{R } =
\support_{\Fcl(R)} = \support_S$. Furthermore, $\Fcl(R)$ can be constructed in polynomial time in
$|X|$. 
\end{theorem}
\begin{proof}  
	We follow here some of the ideas used in the proof of Theorem~16 in~\cite{LAMSH:25}. Let $(R,F)$
	be a pair of relations on $\pairs(X)$, and suppose that $S$ is obtained as described in the
	statement. To simplify notation, put $\mathfrak{R}\coloneqq \mathfrak{R}_{R,F}$ and $\Fcl\coloneqq
	\Fcl(R)$. We first show that $\Fcl \subseteq S$. Since the rules \axiom{R1}-\axiom{R4} are
	exhaustively applied, the final relation $S$ must be $\support^+_{R }$-reflexive, transitive,
	cross-consistent, and $F$-csym and satisfies, by definition, $R\subseteq S$. Thus, $S \in
	\mathfrak{R}$ holds. This together with $\Fcl = \cap_{R' \in \mathfrak{R}} R'$ implies $\Fcl
	\subseteq S$. 

	We now show $S \subseteq \Fcl$ by considering the sequence $R = S_0, S_1, \dots , S_n = S$ of
	relations on $\pairs(X)$, where $S_{i+1}$ is constructed from $S_i$ by application of exactly one
	of the rules \axiom{R1}, \axiom{R2}, \axiom{R3}, or \axiom{R4} for all $i \in \{0, \dots, n-1\}$. 

	Note that all relations $R'\in \mathfrak{R}$ satisfy $R\subseteq R'$ and it, thus, holds that $R
	\subseteq \bigcap_{R' \in \mathfrak{R}} R' = \Fcl$. Hence, $R = S_0 \subseteq \Fcl$ holds. The
	relation $S_1$ is obtained from $S_0$ by applying \axiom{R1}, i.e., $S_1 \coloneqq S_0 \cup
	\{(p,p) \mid p \in \support^+_{R }\}$. The latter equation together with $R\subseteq \Fcl$ and the
	fact that $\Fcl$ is $\support^+_{R }$-reflexive, implies that $S_1 \subseteq \Fcl$.

	Now suppose $S_i \subseteq \Fcl$ for some fixed $i$ with $1 \leq i \leq n-1$. The relation
	$S_{i+1}$ is constructed from $S_i$ by applying one of the rules \axiom{R2}, \axiom{R3}, or
	\axiom{R4}, and we distinguish these three cases. If \axiom{R2} or \axiom{R3} were applied to
	$S_i$, then $S_{i+1} \subseteq \Fcl$ by the same argument as given in the proof of Theorem~16 in
	\cite{LAMSH:25}. Suppose now that $S_{i+1}$ is obtained from $S_i$ by applying \axiom{R4}. By
	construction, $S_{i+1} = S_i \cup \{(q,p)\}$ holds for some $(p,q) \in F \cap S_i$. By assumption,
	$S_i \subseteq \Fcl$ holds and, thus, $(p,q) \in \Fcl$. Hence, $(p,q) \in R'$ holds for all $R'
	\in \mathfrak{R}$. Since every $R' \in \mathfrak{R}$ is $F$-csym and $(p,q) \in F $, we can
	conclude that $(q,p) \in R'$ for all $R' \in \mathfrak{R}$. Hence, $(q,p) \in \Fcl$ and,
	therefore, $S_{i+1} \subseteq \Fcl$. By induction, $S= S_n \subseteq \Fcl$ holds. This together
	with $\Fcl \subseteq S$ implies that $S = \Fcl$. 

	Therefore, we have, in particular, $\support_S = \support_{\Fcl}$. Furthermore, observe that
	$\support_{S_1} = \support^+_{R }$ and by construction, $\support_{S_i} = \support_{S_{i+1}}$ for
	all $i \in \{1, \dots , n-1\}$. Consequently, $\support^+_{R } = \support_S$ holds. 

	Finally, we show that $\Fcl(R)$ can be constructed in polynomial time in $|X|$. By the arguments
	above, $\Fcl(R)$ can be obtained by repeatedly applying rules \axiom{R1}, \axiom{R2}, \axiom{R3},
	and \axiom{R4} starting with $S = R$ and extending it step by step. Since $S,F \subseteq \pairs(X)
	\times \pairs(X)$ and $|\pairs(X)| = \binom{|X|}{2} + |X| \in O(|X|^2)$, it follows that $|S| \in
	O(|X|^4)$ at each step and $|F| \in O(|X|^4)$. Thus, checking whether one of the rules \axiom{R1},
	\axiom{R2}, \axiom{R3}, or \axiom{R4} can be applied to $S$ can be done in polynomial time, and
	adding a pair $(p,q)$ to $S$ requires only constant time. Clearly, the process of constructing
	$\Fcl(R)$ by applying rules \axiom{R1}, \axiom{R2}, \axiom{R3}, or \axiom{R4} must terminate,
	since at most $O(|X|^4)$ pairs can be added to $S$. In summary, $\Fcl(R)$ can be constructed in
	polynomial time in $|X|$. 
\end{proof}

\begin{figure}[h]
    \centering
    \includegraphics[width=0.8\linewidth]{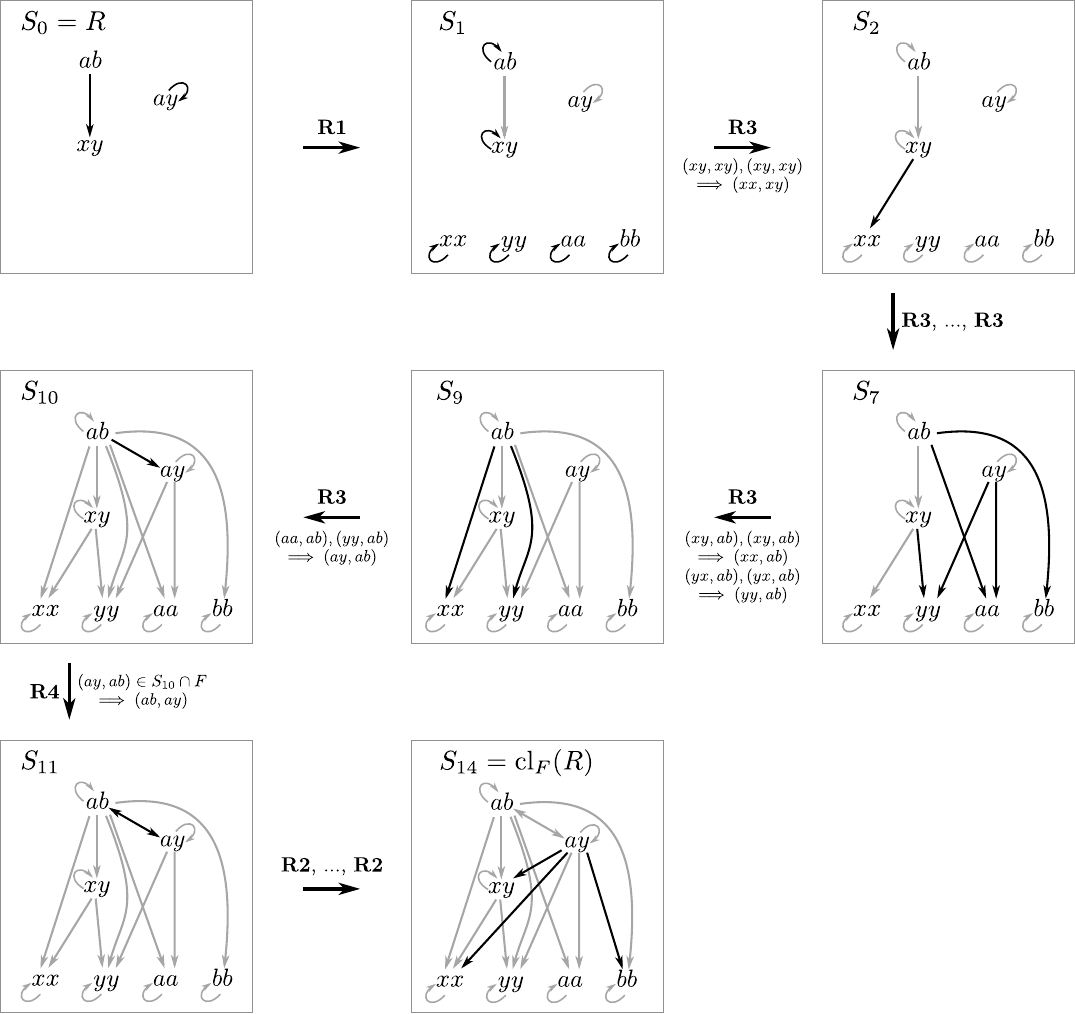}
    \caption{
    Illustrated is the stepwise construction $\Fcl(R)$ of the two relations $R =
    \{(xy,ab),(ay,ay)\}$ and $F = \{(ay,ab)\}$ on $\pairs(X)$ with $X = \{a,b,x,y\}$. For each of
    the shown relations $S_i$, an arc $p\to q$ is drawn precisely if $(q,p) \in S_i$. We start with
    $S_0 = R$ and obtain $\Fcl(R) = S_{14}$ by repeated application of \axiom{R1}, \axiom{R2},
    \axiom{R3}, and \axiom{R4} as indicated by thick arcs. For better visibility, arcs that are
    already contained in $S_{j}$ are shown in gray in $S_i$, $j<i$.}
    \label{fig:working-ex0}
\end{figure}

The stepwise construction of $\Fcl(R)$ for a pair of relations $(R,F)$ using the rules
\axiom{R1}--\axiom{R4}, as described in Theorem~\ref{thm:polynomial_construction_Fclosure}, is
illustrated in Figure~\ref{fig:working-ex0}. We emphasize that our main aim here is to establish
the existence of a polynomial-time construction. We strongly suspect that the procedure described in
the proof can be improved considerably and that sharper runtime bounds can be obtained; developing
such improvements is left for future work. For now,
Theorem~\ref{thm:polynomial_construction_Fclosure} provides the main ingredient needed to show that
$\Fcl$ is indeed a closure operator.

\begin{proposition}\label{prop:closure_operator}
Let $F$ be a relation on $\pairs(X)$. Then, $\Fcl$ is a closure operator, i.e., it satisfies the following three conditions
for all relations $R$ on $\pairs(X)$:
\begin{enumerate}
\item[(i)] Extensivity: $R \subseteq \Fcl(R)$.
\item[(ii)] Monotonicity: $\Fcl(R') \subseteq \Fcl(R)$ for all relations $R'$ on $\pairs(X)$ with $R'\subseteq R$.
\item[(iii)] Idempotency: $\Fcl(\Fcl(R)) =  \Fcl(R)$.
\end{enumerate}
\end{proposition}
\begin{proof}
	Let $R$ and $F$ be relations on $\pairs(X)$. Note that all relations $R'\in \mathfrak{R}_{R,F}$
	satisfy $R\subseteq R'$ and it, thus, holds that $R \subseteq \bigcap_{R' \in \mathfrak{R}_{R,F}}
	R' = \Fcl(R)$. Hence, extensivity holds.

	Suppose now that $R'$ and $R$ are relations on $\pairs(X)$ such that $R'\subseteq R$. Hence,
	$\support_{R'} \subseteq \support_{R}$ holds. Since $R, R'$ are both relations on $\pairs(X)$,
	also $\support_{R'}^+ \subseteq \support_{R}^+$ is satisfied. Now, let $\tilde R \in
	\mathfrak{R}_{R,F}$. Since $\support_{R'}^+ \subseteq \support_{R}^+$ and since $\tilde R$ is
	$\support_{R}^+$-reflexive, it follows that $\tilde R$ is, in particular,
	$\support_{R'}^+$-reflexive. Moreover, since $\tilde R \in \mathfrak{R}_{R,F}$, the relation
	$\tilde R$ is, in addition, transitive, cross-consistent, $F$-csym, and satisfies $R'\subseteq
	R\subseteq\tilde R$. Taking the latter two arguments together, $\tilde R \in \mathfrak{R}_{R',F}$
	follows. Consequently, $\mathfrak{R}_{R,F} \subseteq \mathfrak{R}_{R',F}$ holds and we obtain
	$\Fcl(R') = \bigcap_{\tilde R \in \mathfrak{R}_{R',F}} \tilde R \subseteq \bigcap_{\tilde R \in
	\mathfrak{R}_{R,F}} \tilde R = \Fcl(R)$. Hence, monotonicity holds.    

	To prove idempotency, note that, by Theorem~\ref{thm:polynomial_construction_Fclosure} and
	\axiom{R1}, $\Fcl(R)$ is $\support^+_{R}$-reflexive and $\support^+_{R} = \support_{\Fcl(R)} =
	\support^+_{\Fcl(R)}$. Therefore, $\Fcl(R)$ is $\support_{\Fcl(R)}^+$-reflexive. Moreover, by
	Theorem~\ref{thm:polynomial_construction_Fclosure}, to obtain $\Fcl(R)$ from $R$ the rules
	\axiom{R1}, \axiom{R2}, \axiom{R3}, and \axiom{R4} are exhaustively applied until none of these
	rules can be applied anymore. Consequently, $\Fcl(R)$ is transitive, cross-consistent, $F$-csym,
	and clearly $\Fcl(R) \subseteq \Fcl(R)$. In summary, $\Fcl(R)\in \mathfrak{R}_{\Fcl(R),F}$ holds.
	Hence, $\Fcl(\Fcl(R)) = \bigcap_{\tilde R \in \mathfrak{R}_{\Fcl(R),F}} \tilde R \subseteq
	\Fcl(R)$. Moreover, by extensivity, $\Fcl(R) \subseteq \Fcl(\Fcl(R))$ holds, and it follows that
	$\Fcl(\Fcl(R)) = \Fcl(R)$, i.e., idempotency holds.  
\end{proof}

Theorem~\ref{thm:polynomial_construction_Fclosure} implies that $\Fcl(R)$ is transitive and, thus,
$\tc(\Fcl(R)) = \Fcl(R)$. This together with $R \subseteq \Fcl(R)$ and monotonicity of $\tc$ 
implies that $\tc(R) \subseteq \Fcl(R)$. 
Moreover, Theorem~\ref{thm:polynomial_construction_Fclosure} also implies that $\Fcl(R)$ is $\support^+_R$-reflexive.
For later reference we provide 
\begin{observation} \label{obs:Fclosure_transitive_and_reflexive}
For all pairs $(R,F)$ of relations on $\pairs(X)$, $\Fcl(R)$ is $\support^+_R$-reflexive and transitive and 
it holds that $\tc(R) \subseteq \Fcl(R)$. 
\end{observation}

The closure $\Fcl(R)$ also provides information about every DAG that
\AF-realizes $(R,\FR)$. More precisely, the order relation
$\lca_G(ab) \preceq_G \lca_G(xy)$ holds not only for the constraints
$(ab,xy) \in R$, but for all constraints contained in $\Fcl(R)$. Thus,
$\Fcl(R)$ indeed captures LCA-constraints that are implied by the required
and forbidden constraints in $R$ and $F$, respectively.

\begin{lemma} \label{lemma:Fclosure_preceq}
Let $(R,F)$ be a pair of relations on $\pairs(X)$ and let $G$ be a DAG on $X$ that \AF-realizes
$(R,\FR)$. Then, for all $a,b,x,y \in X$, it holds that
\[
(ab,xy) \in \Fcl(R) \implies \lca_G(ab) \preceq_G \lca_G(xy). 
\]
In particular, it holds that $\Fcl(R)\subseteq \rel_G$.
\end{lemma}
\begin{proof}
We follow here some of the ideas used in the proof of Lemma~20
in~\cite{LAMSH:25}. Suppose that
$(R,F)$ is a pair of relations on $\pairs(X)$ and that $G$ is a DAG on $X$ that $\AF$-realizes
$(R,\FR)$. Let $R = S_0, S_1, \dots , S_n = \Fcl(R)$ be a sequence of relations, where $S_{i+1}$ is
obtained from $S_i$ by applying one of the rules \axiom{R1}, \axiom{R2}, \axiom{R3}, or \axiom{R4}
for all $i \in \{0, \dots, n-1\}$ (w.r.t.\ the two relations $S_i$ and $F$ as specified in
Definition~\ref{def:rules}). Such a sequence exists by
Theorem~\ref{thm:polynomial_construction_Fclosure}. We now prove that, for all $i \in \{0, \dots
,n\}$, it holds that 
\begin{align} \label{imp:Fclosure_preceq}
(ab,xy) \in S_i \implies \lca_G(ab) \preceq_G \lca_G(xy).
\end{align}
For $S_0 = R$, let $(ab,xy) \in S_0$. Since $G$ \AF-realizes $(R,\FR)$ and $(ab,xy) \in R$, we can
conclude that $\lca_G(ab)$ and $\lca_G(xy)$ are well-defined. Thus, one of the axioms \axiom{I1} and
\axiom{I2} must be satisfied, and we obtain $\lca_G(ab) \prec_G \lca_G(xy)$ and $\lca_G(ab) =
\lca_G(xy)$, respectively. In both case, it holds that $\lca_G(ab) \preceq_G \lca_G(xy)$ as desired.
To obtain $S_1$, we apply \axiom{R1} to $S_0$, i.e., $S_1 = S_0 \cup \{(p,p) \mid p \in
\support^+_{R }\}$. Clearly, all elements $(p,q)\in S_0\cap S_1$ satisfy $\lca_G(p)\preceq_G
\lca_G(q)$. Moreover, all $(p,q)\in S_1 \setminus S_0$ must satisfy $p=q$ and trivially $\lca_G(p) =
\lca_G(q)$ holds. Thus, Eq.~\ref{imp:Fclosure_preceq} is satisfied for $S_0$ and $S_1$.

Suppose that Eq.~\ref{imp:Fclosure_preceq} holds for $S_i$ for some fixed $i$ with $1 \leq i \leq
n-1$. By construction, $S_{i+1} = S_i \cup \{(p,q)\}$ for some $p,q \in \support^+_R$. Since
Eq.~\ref{imp:Fclosure_preceq} is correct for $S_i$, it suffices to show that
Eq.~\ref{imp:Fclosure_preceq} holds for $(p,q) \in S_{i+1}$ and we distinguish three cases for the
rules that were used to add $(p,q)$. In case of \axiom{R2} and \axiom{R3}, the correctness follows
by the same arguments used in the proof of Lemma~20 in~\cite{LAMSH:25}. Suppose now that rule
\axiom{R4} was applied to obtain $S_{i+1} = S_i \cup \{(p,q)\}$. Hence, $(q,p) \in F \cap S_i$
holds. Since $(q,p) \in S_i$, Eq.~\ref{imp:Fclosure_preceq} implies that $\lca_G(q) \preceq_G
\lca_G(p)$. Moreover, since $(q,p) \in S_i$, it follows by the same arguments as in the proof of
Theorem~\ref{thm:polynomial_construction_Fclosure} that $q,p \in \support^+_R$. This and $(q,p) \in
F$ implies that $(q,p) \in \FR$. Hence, $\lca_G(q) \nprec_G \lca_G(p)$, as $G$ \AF-realizes
$(R,\FR)$. Taking the latter two arguments together, $\lca_G(q) = \lca_G(p)$ must hold. Therefore,
Eq.~\ref{imp:Fclosure_preceq} is satisfied for $(p,q) \in S_{i+1}$. Since $S_{i} = S_{i+1}\setminus
(p,q)$ and since, by induction assumption, Eq.~\ref{imp:Fclosure_preceq} holds for $S_i$ it follows
that Eq.~\ref{imp:Fclosure_preceq} is satisfied for all elements in $S_{i+1}$.

In summary, Eq.~\ref{imp:Fclosure_preceq} holds for all elements in $S_i$ with $i \in \{0, \dots ,
n\}$, and, thus, in particular for all elements in $S_n = \Fcl(R)$. The latter arguments together
with the definition of $\rel_G$ imply $\Fcl(R)\subseteq \rel_G$.
\end{proof}

\subsection{Characterization of \AF-Realizability}
\label{subs:charAFreal}

We are now ready to introduce two conditions, \axiom{Y1} and \axiom{Y2}, that
will characterize \AF-realizable pairs of relations. Their necessity follows
readily from the preceding results; see Lemma~\ref{lemma:AF_realiz_implies_Y1_Y2}.
To prove sufficiency, we construct a specific DAG and show that it
\AF-realizes $(R,F)$ whenever \axiom{Y1} and \axiom{Y2} are satisfied. We
begin by defining these two conditions.

\begin{definition}[Condition \axiom{Y1} and \axiom{Y2}]
For a pair $(R,F)$ of relations on $\pairs(X)$, we define the 
following two conditions: 
\begin{description}
\item[\axiom{Y1}]  For all $a,b,x\in X$: $ab\neq xx$ implies 
$(ab, xx)\notin \Fcl(R)$.
\end{description}
\begin{description}
\item[\axiom{Y2}]  For all $a,b,x,y\in X$: 
$(ab, xy)\in R$ and  $(xy, ab)\not\in \tc(R)$ implies $(xy,ab)\notin \Fcl(R)$. 
\end{description}
\end{definition}

Note that \axiom{Y1} and \axiom{Y2} form a natural generalization of Conditions \axiom{X1} and
\axiom{X2}, which characterize realizable relations. \axiom{Y1} differs from \axiom{X1} by the
condition $(ab, xx)\notin \Fcl(R)$ instead of $(ab, xx)\notin R$. Moreover, \axiom{Y2} differs from
\axiom{X2} in the chosen closures $\Fcl(R)$ and $\cl(R)$. In particular, all these conditions
coincide whenever $F = \emptyset$, as shown next. 

\begin{lemma}\label{lem:XY-eq-emptyF}
For a pair $(R,\emptyset)$ of relations on $\pairs(X)$ it holds that:
\begin{enumerate}
    \item[(1)] $(R,\emptyset)$ satisfies \axiom{Y1} if and only if $R$ satisfies \axiom{X1}.
    \item[(2)] $(R,\emptyset)$ satisfies \axiom{Y2} if and only if $R$ satisfies \axiom{X2}.
\end{enumerate}
\end{lemma}
\begin{proof}
	Consider a pair $(R,\emptyset)$ of relations on $\pairs(X)$. As shown in \cite[L~23]{LAMSH:25},
	$R$ satisfies \axiom{X1} if and only if $\cl(R)$ satisfies \axiom{X1}. This combined with
	Observation~\ref{obs:cl_emptyset=cl} implies that $R$ satisfies \axiom{X1} if and only if
	$\cl_{\emptyset}(R)$ satisfies \axiom{X1} or, equivalently, if $(R,\emptyset)$ satisfies
	\axiom{Y1}. Hence, Statement (1) holds. Statement (2) is an immediate consequence of
	Observation~\ref{obs:cl_emptyset=cl} and the definition of \axiom{Y2} and \axiom{X2}.
\end{proof}

To provide a characterization of \AF-realizable relations $(R,F)$, we first show that those
relations satisfy \axiom{Y1} and \axiom{Y2} which generalizes \cite[L~24]{LAMSH:25}.

\begin{lemma} \label{lemma:AF_realiz_implies_Y1_Y2}
	If a pair $(R,F)$ of relations on $\pairs(X)$ is \AF-realizable, then $(R,F)$ satisfies \axiom{Y1}
	and \axiom{Y2}.
\end{lemma}
\begin{proof} 
	Let $(R,F)$ be a pair of relations on $\pairs(X)$ which is \AF-realized by a DAG $G$ on $X$. By
	Proposition~\ref{prop:RF_real_iff_RFR_real}, $G$ \AF-realizes $(R,\FR)$. Assume, for contradiction,
	that \axiom{Y1} is not satisfied for $(R,F)$. Hence, there exists $a,b,x \in X$ with $ab \neq xx$
	such that $(ab, xx) \in \Fcl(R)$. This together with Lemma~\ref{lemma:Fclosure_preceq} implies
	$\lca_G(ab) \preceq_G \lca_G(xx)$. Since $x$ is a leaf, $\lca_G(ab) = \lca_G(xx) = x$ must hold.
	Hence, $a = b = x$; a contradiction to $ab \neq xx$. 

	To show that \axiom{Y2} is satisfied, suppose there are $a,b,x,y\in X$ such that $(ab, xy)\in R$ and
	$(xy, ab)\not\in \tc(R)$. This and the fact that $G$ realizes $R$ together with \axiom{I1} implies
	that $\lca_G(ab) \prec_G \lca_G(xy)$. Contraposition of Lemma~\ref{lemma:Fclosure_preceq} thus
	implies that $(xy, ab) \notin \Fcl(R)$. Hence, \axiom{Y2} holds.
\end{proof}

To prove that \axiom{Y1} and \axiom{Y2} in fact characterize \AF-realizable relations $(R,F)$, we
construct the \emph{canonical DAG} $\GG_{R,F}$ based on properties of $\Fcl(R)$. 
The construction has a simple interpretation: elements $ab$ and $cd$ that occur in 
symmetric pairs $(ab,cd)$ and $(cd,ab)$ in
$\Fcl(R)$ must have the same LCA and are therefore represented by one vertex; the remaining order
between these elementsdetermines the arcs of the canonical DAG.
As we shall see in
Theorem~\ref{thm:characterization_AF_realized}, if $(R,F)$ is \AF-realizable, then a slightly
modified version of $\GG_{R,F}$ \AF-realizes $(R,F)$. The idea of a canonical DAG derived from
properties of an underlying closure is not new, and we largely follow the approach
of~\cite{LAMSH:25}.

For the definition of the canonical DAG associated with $(R,F)$, we use a standard procedure that
associates a poset with a reflexive and transitive relation (also known as a \emph{preorder} or
\emph{quasi-order}) by identifying precisely those elements that violate anti-symmetry. 
By Observation~\ref{obs:Fclosure_transitive_and_reflexive}, $\Fcl(R)$ is a $\support_R^+$-reflexive and transitive relation on $\pairs(X)$. 
Since $\support_{\Fcl(R)} = \support^+_R$ by Theorem~\ref{thm:polynomial_construction_Fclosure}, 
we can consider $\Fcl(R)$ as a relation on $\support^+_R$. 
In this case, $\Fcl(R)$ is reflexive and transitive and, thus, a preorder. 
To associate now a poset with this initial preorder $\Fcl(R)$, we put two
elements $p$ and $q$ in relation $\sim_{\Fcl(R)}$ whenever both $(p,q)\in \Fcl(R)$ and $(q,p)\in
\Fcl(R)$ hold. This, in turn yields an equivalence relation $\sim_{\Fcl(R)}$ whose equivalence
classes will correspond to the vertices of $\GG_{R,F}$. We then define a partial order
$\leq_{\Fcl(R)}$ on these equivalence classes which is used to define the arcs of $\GG_{R,F}$.

\begin{definition}[{The quotient poset $(Q,\le_{\Fcl(R)})$ of $\Fcl(R)$}]
	Let $(R,F)$ be a pair of relations on $\pairs(X)$. We define the equivalence relation
	$\sim_{\Fcl(R)}$ on $\support^+_R$ by putting, for all $p,q \in \support^+_R$, \[p \sim_{\Fcl(R)}
	q \iff (p,q) \in \Fcl(R) \text{ and } (q,p) \in \Fcl(R). \] Let $[p]$ denote the equivalence class
	of $\sim_{\Fcl(R)}$ that contains $p\in\support_R^+$, and let $Q$ denote the set of all such
	equivalence classes. Define the partial order $\leq_{\Fcl(R)}$ on $Q$ by putting, for all classes
	$[p]$ and $[q]$ in $Q$, \[[p] \leq_{\Fcl(R)} [q] \iff (p,q) \in \Fcl(R).\] We refer to the poset
	$(Q,\le_{\Fcl(R)})$ as the \emph{quotient poset of $\Fcl(R)$}.
\end{definition}

The well-definedness and the facts that $\sim_{\Fcl(R)}$ is an equivalence relation and that
$(Q,\le_{\Fcl(R)})$ is a poset follows from the standard ``quotient'' construction that turns a
preorder into a partial order; see \cite[Prop~5.2.4]{Schroder:03} for a full proof. Based on the
equivalence relation $\sim_{\Fcl(R)}$ and partial order $\leq_{\Fcl(R)}$, we define the canonical
DAG $\GG_{R,F}$, similar to the definition of a canonical DAG given in \cite[Def~26]{LAMSH:25}.

\begin{definition}[Canonical DAG]\label{def:canonDAGN}
Let $(R,F)$ be a pair of relations on $\pairs(X)$ and let $(Q,\leq_{\Fcl(R)})$ be the quotient poset
of $\Fcl(R)$. The \emph{canonical DAG} $\GG_{R,F}$ is defined as the DAG obtained from the Hasse
diagram $\Hasse(Q,\leq_{\Fcl(R)})$ by relabeling all vertices $[aa]\in Q$ with $a$.
\end{definition}

Illustrative examples of the canonical DAG are provided in
Figures~\ref{fig:working-ex1} and~\ref{fig:GrealR-notRF.jpg}; see also
Examples~\ref{exmpl:F=FR} and~\ref{exmpl:Fcsym}.

We emphasize that if $(R,F)$ does not satisfy \axiom{Y1}, then $\cl_F(R)$ may contain constraints of
the form $(ab,xx)$. In this case, it is even possible that, for example, $[aa]=[bb]\in Q$. Hence, in
the construction of $\GG_{R,F}$, the vertex $[aa]$ would have to be relabeled by $a$, while the same
vertex, viewed as $[bb]$, would have to be relabeled by $b$. Thus, the construction of $\GG_{R,F}$
would no longer be well-defined. One can circumvent this ambiguity by choosing, for each such class,
a single representative leaf by which it is replaced (e.g. by imposing a linear order on $X$ and
choosing the representative as the smallest element). Nevertheless, we will only explicitly
construct $\GG_{R,F}$ for pairs $(R,F)$ of relations on $\pairs(X)$ that satisfy \axiom{Y1}. Under
this assumption, the ambiguity described above cannot occur. Consequently, $\GG_{R,F}$ is a
well-defined DAG and, as we shall see in Proposition~\ref{prop:properties_of_canonical_DAG}, has
leaf set $X$.

We further note in passing that $\GG_{R,\emptyset}$ coincides with the canonical DAG $\GG_R$ as
defined in~\cite{LAMSH:25}. Furthermore, observe that the canonical DAG $\GG_{R,F}$ is only based on
properties of the relation $\Fcl(R)$ and similarly, the construction of $\GG_{\Fcl(R),F}$ depends
only on the relation $\Fcl(\Fcl(R))$. By Proposition~\ref{prop:closure_operator}, $\Fcl(\Fcl(R)) =
\Fcl(R)$ and we obtain
\begin{observation}\label{obs:canDAG_eq_canDAG_of_closure}
For all pairs $(R,F)$ of relations on $\pairs(X)$, we have $\GG_{R,F} = \GG_{\Fcl(R),F}$.    
\end{observation}

\begin{figure}
    \centering
    \includegraphics[width=0.6\textwidth]{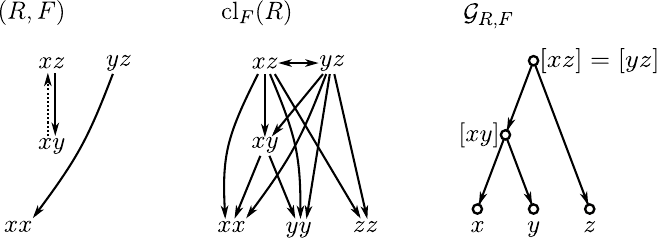}
    \caption{On the left, we give a graphical representation of the relations 
    $R= \{(xy,xz), (xx,yz)\}$ and $F=\{(xz,xy)\}$ on $\pairs(X)$ with $X=\{x,y,z\}$. Here we 
    draw a solid (resp.\ dashed) arc $p\to q$ precisely if $(q,p) \in R$
    (resp.\ $(q,p) \in F$).  In addition, the graphical representation 
    of $\Fcl(R)$ is provided where arcs $(p,p)$  are omitted for all $p\in \support_{\Fcl(R)}$. Furthermore,
    the canonical DAG $\GG_{R,F}$ which \AF-realizes $(R,F)$
    is shown.
    }
    \label{fig:working-ex1}
\end{figure}

We illustrate the definition using the following two examples.
The first example shows that the original three rules may already suffice, whereas the second
demonstrates why the additional $F$-conditional-symmetry rule is necessary.
\begin{example}[\AF-realization]\label{exmpl:F=FR}
Consider the relation $R= \{(xy,xz), (xx,yz)\}$ and $F=\{(xz,xy)\}$ on $\pairs(X)$ with
$X=\{x,y,z\}$, which are presented in Figure~\ref{fig:working-ex1}. In this example, we have
$F=\FR$. It is straightforward to verify that \axiom{R4} is not applied in deriving $\Fcl(R)$. In
other words, only \axiom{R1}, \axiom{R2}, and \axiom{R3} are used. Since these rules are independent
of $F$, it follows that $\Fcl(R) = \cl_{\emptyset}(R)$, i.e., $\GG_{R,\emptyset} = \GG_{R,F}$. 
In this example, $\GG_{R,F}$ \AF-realizes $(R,F)$.   
\end{example}

\begin{example}[Rule \axiom{R4} is indispensable]\label{exmpl:Fcsym}
In contrast to the relations as in Example~\ref{exmpl:F=FR}, the relations $R= \{(xy,xz), (yy,yz)\}$
and $F=\{(yz,xz)\}$ on $\pairs(X)$ with $X=\{x,y,z\}$ as illustrated in
Figure~\ref{fig:GrealR-notRF.jpg} are not \AF-realized by $\GG_{R,\emptyset}$. Here we have $F=\FR$.
In particular, applying \axiom{R1}, \axiom{R2}, and \axiom{R3} yields $\cl_{\emptyset}(R)$, for
which $(yz,xz)\in \cl_{\emptyset}(R)$ but $(xz,yz)\notin \cl_{\emptyset}(R)$. Consequently, the
$\sim_{\cl_{\emptyset}(R)}$-classes $[xz]$ and $[yz]$, which form the vertices of
$\GG_{R,\emptyset}$, are distinct. In particular, as can be seen in
Figure~\ref{fig:GrealR-notRF.jpg}, $[yz] = \lca_{\GG_{R,\emptyset}}(yz) \prec_{\GG_{R,\emptyset}}
\lca_{\GG_{R,\emptyset}}(xz) = [xz]$. Together with $(yz,xz)\in F$, this shows that
$\GG_{R,\emptyset}$ does not satisfy \axiom{F} and, thus, does not \AF-realize $(R,F)$. However,
applying \axiom{R1}--\axiom{R4} yields $(yz,xz), (xz,yz)\in \Fcl(R)$ and, thus, the
$\sim_{\Fcl(R)}$-classes $[xz]$ and $[yz]$ coincide, forming a single vertex in $\GG_{R,F}$. Hence,
we obtain $\lca_{\GG_{R,F}}(xz) = [xz] = [yz] = \lca_{\GG_{R,F}}(yz)$ and $\GG_{R,F}$ \AF-realizes
$(R,F)$. Therefore, in general, \axiom{R4} cannot be omitted in the construction of $\GG_{R,F}$.
\end{example}

\begin{figure}
    \centering
    \includegraphics[width=0.8\textwidth]{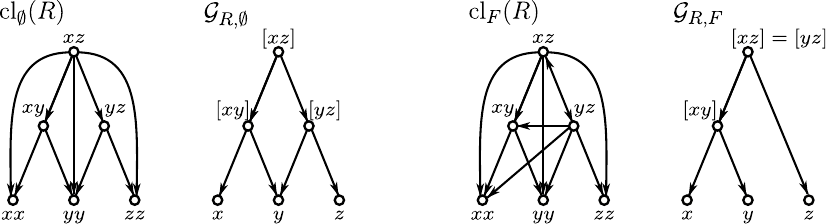}
    \caption{Shown is the graphical representation of $\cl_{\emptyset}(R)$
    and $\Fcl(R)$ for the relations $R = \{(xy,xz),(yy,yz)\}$ and $F=\{(yz,xz)\}$. 
    Here, we draw an arc $p\to q$ precisely if $(q,p) \in \cl$ but 
    omitted arcs of the form $(p,p)$, $\cl\in \{\cl_{\emptyset}(R),\Fcl(R)\}$.
    In addition, the canonical DAGs $\GG_{R,\emptyset}$
    and $\GG_{R,F}$ are shown. The DAG $\GG_{R,\emptyset}$
     \AF-realizes $(R,\emptyset)$
    but not $(R,F)$. Moreover, $\GG_{R,F}$ \AF-realizes $(R,F)$, 
    see Example~\ref{exmpl:Fcsym} for further details.} 
    \label{fig:GrealR-notRF.jpg}
\end{figure}

We emphasize that $\GG_{R,F}$ \AF-realizes $(R,F)$ if $F = \FR$ as it was the case in
Example~\ref{exmpl:F=FR} and \ref{exmpl:Fcsym}. In general, however, $F\neq \FR$ in which case a
slightly modified version (the $\FR$-extension) of $\GG_{R,F}$ must be used to obtain a DAG that
\AF-realizes $(R,F)$ whenever $(R,F)$ is \AF-realizable; see
Theorem~\ref{thm:characterization_AF_realized}.

Before giving one of the main results of this paper, we establish
some properties of the canonical DAG needed for its proof.

\begin{proposition}\label{prop:properties_of_canonical_DAG}
Let $(R,F)$ be a pair of relations on $\pairs(X)$ that satisfies \axiom{Y1} or that is
\AF-realizable. Then, the canonical DAG $\GG_{R,F}$ is a well-defined DAG on $X$ such that $[ab] =
\lca_{\GG_{R,F}}(ab)$ for all vertices $[ab]$ of $\GG_{R,F}$ with $a \neq b$ and $x =
\lca_{\GG_{R,F}}(xx)$ for all $x \in X$. In particular, $\GG_{R,F}$ is 2-lca-relevant, phylogenetic,
and \AF-realizes $(\Fcl(R),\FR)$. Moreover, the $\FR$-extension of $\GG_{R,F}$ is phylogenetic.
\end{proposition}
\begin{proof}
	We closely follow the proof of \cite[Prop~28]{LAMSH:25}. It differs, in essence, only from the
	part showing that $\GG_{R,F}$ \AF-realizes $(\Fcl(R),\FR)$. Let $(R,F)$ be a pair of relations on
	$\pairs(X)$ that satisfies \axiom{Y1} or that is \AF-realizable. In the latter case,
	Lemma~\ref{lemma:AF_realiz_implies_Y1_Y2} implies that $(R,F)$ satisfies \axiom{Y1}. To simplify
	writing, let $\GG \coloneqq \GG_{R,F}$, $\sim \coloneqq \sim_{\Fcl(R)}$, $\leq \coloneqq
	\leq_{\Fcl(R)}$, and $\Fcl\coloneqq \Fcl(R)$. Moreover, let $Q$ be the set of all $\sim$-classes. 

	We start by showing that $\GG$ has leaf set $X$. By \axiom{Y1}, no class $[xx]\in Q$ can coincide
	with a class $[ab]\in Q$ for which $ab\neq xx$. Moreover, by \axiom{R1}, $[xx]\in Q$ for all $x\in
	X$. Hence, the two sets $Q_1=\{[xx]\colon x\in X\}$ and $Q_2 =\{[ab]\in Q\colon a\neq b\}$ form a
	bipartition of $Q$, i.e., $Q =Q_1\cup Q_2$ and $Q_1\cap Q_2 =\emptyset$. By the latter arguments,
	we have $x,y\in X$ and $x\neq y$ if and only if $[xx]\neq [yy]$ and $[xx],[yy]\in Q_1\subseteq Q$.
	Thus, there is a 1-to-1 correspondence between the vertices in $X$ and the classes in $Q_1$.
	Moreover, by \axiom{Y1}, $(ab, xx)\not\in \Fcl$ for all $a,b,x\in X$ with $ab \neq xx $. Hence,
	$[ab] \not \le [xx]$ for all $a,b,x\in X$ with $ab \neq xx$ and by construction, $\GG$ does not
	contain arcs $(v,v)$ for any $v\in V(\GG)$. Therefore, one never adds an arc of the form $x\to v$
	for any $x\in X$ and $v \in V(\GG)$ in $\GG$. In summary, the set $X$ forms a subset of leaves in
	$\GG$. It remains to show that all vertices $[xy]\in Q$ with $x\neq y$ are non-leaf vertices in
	$\GG$. To this end, observe that cross-consistency of $\Fcl$ ensures that $[xx]\leq [xy]$ for all
	classes $[xy]$ in $Q$. Hence, for all such classes $[xy]$ with $x\neq y$, the vertex $x$ satisfies
	$x\prec_{\GG} [xy]$. Thus, classes $[xy]$ of $Q$ with $x\neq y$ must be non-leaf vertices in
	$\GG$. Hence, $X$ is precisely the leaf set of $\GG$. 

	Now let $v\in V(\GG)$. Suppose that $v=[ab]$ for some class $[ab]\in Q$ with $a\neq b$. We show
	that $v$ is the unique least common ancestor of $a$ and $b$ in $\GG$. By the previous arguments,
	$[ab]$ is a common ancestor of $a$ and $b$. Now suppose that $w \in V(\GG)$ is a common ancestor
	of $a$ and $b$ with $w \neq [ab]$. Then $w$ must be in $Q_2$, i.e., $w=[xy]\in Q$ with $x\neq y$.
	Since $w$ is a common ancestor of $a$ and $b$, there are directed paths $[xy]\leadsto a$ and
	$[xy]\leadsto b$ in $\GG$. By construction of $\GG$ and since $\le$ is transitive (as $(Q,\le)$
	is a poset), we must have $[aa]\le [xy]$ and $[bb]\le [xy]$. By definition of $\le$, we
	obtain $(aa,xy) \in \Fcl$ and $(bb,xy) \in \Fcl$. Since $[ab]\in Q$, we have, by construction of
	$\sim$, $ab\in\support_{R}^+$. By Theorem~\ref{thm:polynomial_construction_Fclosure},
	$\support_{R}^+ =\support_{\Fcl}$ and, therefore, $ab\in\support_{\Fcl}$ holds. This together with
	the fact that $\Fcl$ is cross-consistent implies that $(ab,xy) \in \Fcl$. Again, by definition of
	$\le$, we have $[ab]\le [xy] =w$ and we obtain $[ab]\preceq_\GG w$ by definition of $\GG$. Since
	$w=[xy]$ was an arbitrary common ancestor of $a$ and $b$, $[ab] \preceq_\GG w$ holds for
	\emph{all} common ancestors $w$ of $a$ and $b$. Hence, $[ab]$ is the \emph{unique} least common
	ancestor of $a$ and $b$. Moreover, if $v=a$ for some $a\in X$, the unique least common ancestor
	$\lca_\GG(aa)$ is $a$, since $a\in X$ is a leaf. In summary, $\GG$ is a DAG in which
	$\lca_\GG(aa)=a$ for all $a\in X$ and $\lca_\GG(ab)=[ab]$ for all $[ab]\in Q$ with $a\neq b$. 

	By the latter arguments and since each vertex in $\GG$ is either of the form $[ab]\in Q_2$ or $x$
	for some $[xx]\in Q_1$, it follows that, for each vertex $v$ in $\GG$, there are $x,y\in X$ such
	the $v=\lca_\GG(xy)$. Therefore, $\GG$ is 2-lca-relevant. This observation allows us to apply
	\cite[L~3.10]{HL:24} to conclude that $\GG$ is phylogenetic. Moreover, by
	Observation~\ref{obs:FR-extension}, the $\FR$-extension $G$ of $\GG$ is based on $xy$-extensions
	for all $xy \in \support_F \setminus \support^+_R$ with $x\neq y$. In particular, any
	$xy$-extension yields two new vertices $u$ and $v$ that have out-degree two. It is now an easy
	task to verify that $G$ remains phylogenetic.

	We show now that $\GG$ realizes $\Fcl$. By definition, $\Fcl$ is transitive and thus,
	$\tc(\Fcl)=\Fcl$. Let $(ab,cd)\in \Fcl$. Suppose that $(cd,ab)\notin \tc(\Fcl)=\Fcl$. Hence,
	$[ab]\neq [cd]$ and $[ab]\le [cd]$. If $a\neq b$, the vertex $v = [ab]$ exists in $\GG$ and if
	$a=b$, the vertex $v=a\in X$ exists in $\GG$. Either way, $[ab]\neq [cd]$, $[ab]\leq [cd]$, and
	the construction of $\GG$ ensures that $\lca_{\GG}(ab)=v\prec_{\GG}[cd]=\lca_{\GG}(cd)$. In other
	words, \axiom{I1} holds for $\GG$ and $\Fcl$. If $(cd,ab)\in \tc(\Fcl)=\Fcl$, then $ab\sim cd$
	and, thus, $[ab]=[cd]$. This together with the latter arguments implies that $\lca_\GG(ab) =
	\lca_\GG(cd)$ holds. Thus, \axiom{I2} holds for $\GG$ and $\Fcl$. In summary, $\GG$ realizes
	$\Fcl$.

	It remains to argue that $\GG$ satisfies \axiom{F} w.r.t.\ $\FR$. To this end, suppose that
	$(ab,cd) \in \FR$. Note that $\sim$ is a relation on $\support^+_{R}$ and $ab,cd \in
	\support^+_{R}$. Thus, $\lca_\GG(ab)$ and $\lca_\GG(cd)$ are well-defined. Now assume for
	contradiction, that $\lca_\GG(ab) \prec_\GG \lca_\GG(cd)$. Hence, there exists a directed path
	$[cd] = w_0 \to w_1 \to \dots \to w_k = v$ in $\GG$, where we define $v\coloneqq a$ if $a = b$ and
	otherwise, $v\coloneqq [ab]$. Note that each vertex $w_i$ corresponds to a $\sim$-class $[p_i]$.
	In particular, for each arc $w_i \to w_{i+1}$, $0 \leq i \leq k-1$, it holds, by construction of
	$\GG$, that there are $xy\in [p_i] $ and $x'y'\in [p_{i+1}]$ such that $(x'y',xy) \in \Fcl$. This
	together with the fact that $\Fcl$ is transitive implies that $(ab,cd) \in \Fcl$. Since $\Fcl$ is
	$F$-csym and $(ab,cd) \in \FR \cap \Fcl$, it holds that $(cd,ab) \in \Fcl$. As shown in the
	previous paragraph, in this case $\lca_\GG(ab) = \lca_\GG(cd)$ holds; a contradiction to
	$\lca_\GG(ab) \prec_\GG \lca_\GG(cd)$. Hence, \axiom{F} is satisfied and $\GG$ \AF-realizes
	$(\Fcl, \FR)$.
\end{proof}

We are now in the position to provide one of the main results of this contribution, namely a
characterization of \AF-realizable pairs $(R,F)$. This result generalizes and extends
\cite[Thm~31]{LAMSH:25}. It links three complementary viewpoints: realizability itself, the two
closure conditions \axiom{Y1} and \axiom{Y2}, and an explicit realization obtained from the
canonical DAG and its $\FR$-extension.

\begin{theorem}\label{thm:characterization_AF_realized}
For a pair $(R,F)$ of relations on $\pairs(X)$ the following statements are equivalent: 
\begin{enumerate}[label={(\arabic*)}]
    \item $(R,F)$ is \AF-realizable. 
    \item $(R,F)$ satisfies \axiom{Y1} and \axiom{Y2}. 
    \item $(R,\FR)$ is \AF-realized by its canonical DAG $\GG_{R,F}$.
    \item $(R,F)$ is \AF-realized by the $\FR$-extension of its canonical DAG $\GG_{R,F}$.
    \item $(R,F)$ is \AF-realized by the phylogenetic network $N$ obtained from the $\FR$-extension of $\GG_{R,F}$ according to Lemma~\ref{lem:DAG2Network}.
    \item $(R,F)$ satisfies \axiom{Y2} and $\Fcl(R)$ is realizable.
\end{enumerate}
\end{theorem}
\begin{proof}
	Let $(R,F)$ be a pair of relations on $\pairs(X)$. By Lemma~\ref{lemma:AF_realiz_implies_Y1_Y2},
	Statement~(1) implies Statement~(2). 

	Now suppose that Statement~(2) is satisfied. Let $\GG \coloneqq \GG_{R,F}$ be the canonical DAG of
	$(R,F)$. Since $(R,F)$ satisfies \axiom{Y1}, Proposition~\ref{prop:properties_of_canonical_DAG}
	implies that $\GG$ \AF-realizes $(\Fcl(R),\FR)$. To show that $\GG$ realizes $R$, we must verify
	\axiom{I1} and \axiom{I2}. To this end, suppose that $(ab,xy) \in R \subseteq \Fcl(R)$. Assume
	first that $(xy,ab) \notin \tc(R)$. Since $(R,F)$ satisfies \axiom{Y2}, $(xy,ab) \notin \Fcl(R)$
	holds. Since $(ab,xy) \in \Fcl(R)$ and $(xy,ab) \notin \Fcl(R) = \tc(\Fcl(R))$ and $\GG$ and
	$\Fcl(R)$ satisfy \axiom{I1}, $\lca_\GG(ab) \prec_\GG \lca_\GG(xy)$ holds. Therefore, $\GG$ and
	$R$ satisfy \axiom{I1}. Assume now that $(xy,ab) \in \tc(R)$. Since $\tc(R) \subseteq \Fcl(R)$ by
	Observation~\ref{obs:Fclosure_transitive_and_reflexive}, we have $(ab,xy), (xy,ab) \in \Fcl(R) = \tc(\Fcl(R))$.
	Since $\GG$ and $\Fcl(R)$ satisfy \axiom{I2}, $\lca_\GG(ab) = \lca_\GG(xy)$ holds. Hence,
	\axiom{I2} is satisfied, and $\GG$ realizes $R$. Moreover, since $\GG$ \AF-realizes $(\Fcl(R),
	\FR)$, it follows that \axiom{F} holds. In summary, $\GG$ \AF-realizes $(R,\FR)$. Thus, Statement
	(2) implies (3). Moreover, Statement (3) implies (4) by
	Proposition~\ref{prop:RF_real_iff_RFR_real}.

	We now show that Statement (4) implies Statement (5). Suppose $(R,F)$ is \AF-realized by the
	$\FR$-extension $G$ of its canonical DAG $\GG_{R,F}$. Note that, since $(R,F)$ is \AF-realizable,
	Proposition~\ref{prop:properties_of_canonical_DAG} implies that $\GG_{R,F}$ and, therefore, $G$ is
	well-defined. Let $N$ be the network obtained from $G$ by Lemma~\ref{lem:DAG2Network} and denote
	its unique root by $\rho$. By Proposition~\ref{prop:properties_of_canonical_DAG}, the DAG $G$ is
	phylogenetic and, thus, $N$ is phylogenetic by Lemma~\ref{lem:DAG2Network}. It remains to show
	that $N$ \AF-realizes $(R,F)$. Suppose $(ab,cd) \in R$. Since $G$ realizes $R$, the vertices
	$\lca_G(ab)$ and $\lca_G(cd)$ are well-defined. By Lemma~\ref{lem:DAG2Network}, $\lca_G(ab) =
	\lca_N(ab)$ and $\lca_G(cd) = \lca_N(cd)$ follows as well as $N$ satisfying \axiom{I1} respectively
	\axiom{I2} for $(ab,cd)$. Hence, $N$ realizes $R$. Now suppose $(ab,cd) \in F$. Assume first that
	$(ab,cd) \in \FR$. Thus, $ab,cd \in \support^+_R$ and since $G$ \AF-realizes $(R,F)$, the vertices
	$\lca_G(ab)$ and $\lca_G(cd)$ are well-defined and satisfy $\lca_G(ab) \nprec_G \lca_G(cd)$. By
	Lemma~\ref{lem:DAG2Network}, $\lca_N(ab) \nprec_N \lca_N(cd)$ holds and $N$ satisfies \axiom{F}
	for $(ab,cd)$. Now suppose that $(ab,cd) \in F \setminus \FR$. By construction of the
	$\FR$-extension, $|\LCA_G(ab)|>1$ or $|\LCA_G(cd)|>1$ holds. Suppose w.l.o.g.\ that
	$|\LCA_G(ab)|>1$. By Lemma~\ref{lem:DAG2Network}, $\LCA_G(ab) = \LCA_N(ab)$ holds. Thus,
	$\lca_N(ab)$ is not well-defined and $N$ satisfies \axiom{F} for $(ab,cd)$. In summary, $N$
	\AF-realizes $(R,F)$.

	Statement (5) clearly implies (1). To summarize so far, Statement (1), (2), (3), (4), and (5) are
	equivalent. 

	To finish the proof, we show that Statement (2) and (6) are equivalent. Suppose Statement (2)
	holds and, thus, that $(R,F)$ satisfies in particular \axiom{Y2}. By
	Proposition~\ref{prop:properties_of_canonical_DAG} and \axiom{Y1}, $(\Fcl(R),\FR)$ is
	\AF-realizable and, thus, $\Fcl(R)$ is realizable. Hence, Statement (2) implies (6). Suppose now
	that Statement (6) holds. Thus, $(R,F)$ satisfies \axiom{Y2} and $\Fcl(R)$ is realizable. It
	remains to show that $(R,F)$ satisfies \axiom{Y1}. Since $\Fcl(R)$ is realizable, it follows that
	the pair $(\Fcl(R),\emptyset)$ is \AF-realizable. Lemma~\ref{lemma:AF_realiz_implies_Y1_Y2}
	implies that $(\Fcl(R),\emptyset)$ satisfies \axiom{Y1}. Hence, for all $a,b,x\in X$ it holds that
	$ab\neq xx$ implies $(ab, xx)\notin \cl_\emptyset(\Fcl(R))$. Since $\Fcl(R) \subseteq
	\cl_\emptyset(\Fcl(R))$ by Proposition~\ref{prop:closure_operator}, $(ab, xx)\notin \Fcl(R)$ and
	$(R,F)$ satisfies \axiom{Y1}. Therefore, Statement (6) implies (2), which completes this proof. 
\end{proof}

Theorem~\ref{thm:characterization_AF_realized} is both structural and constructive. Conditions
\axiom{Y1} and \axiom{Y2} provide the recognition criterion, while the canonical DAG and the
$\FR$-extension provide an explicit witness whenever the criterion is satisfied.

Before showing that the realizability problem can be solved in polynomial time in $|X|$, we provide
two examples to illustrate $\AF$-realizability.

\begin{figure}
    \centering
    \includegraphics[width=0.4\textwidth]{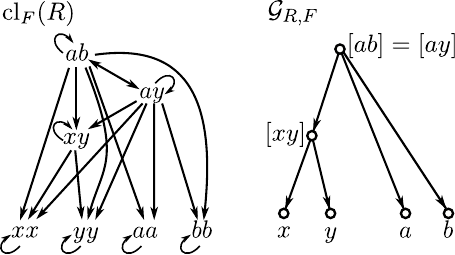}
    \caption{ 
    Shown is the closure $\Fcl(R)$ as constructed in Figure~\ref{fig:working-ex0} for the
    two relations $R = \{(xy,ab),(ay,ay)\}$ and $F = \{(ay,ab)\}$ on $\pairs(X)$ with $X =
    \{a,b,x,y\}$. Here, an arc $p\to q$ is drawn precisely if $(q,p) \in \Fcl(R)$. In
    addition, the canonical DAG $\GG_{R,F}$ which $\AF$-realizes $(R,F)$ is
    shown. For more details, see Example~\ref{exmpl:Y1andY2}.}
    \label{fig:Y1andY2}
\end{figure}

\begin{example}[\axiom{Y1} and \axiom{Y2}]\label{exmpl:Y1andY2}
Consider the two relations $R = \{(xy,ab),(ay,ay)\}$ and $F = \{(ay,ab)\}$ on $\pairs(X)$ with $X =
\{a,b,x,y\}$ and see Figure~\ref{fig:Y1andY2} for $\Fcl(R)$ and the canonical DAG $\GG_{R,F}$. Note
that the stepwise construction of $\Fcl(R)$ was discussed in Figure~\ref{fig:working-ex0}. It is now
easy to verify that $(R,F)$ satisfies \axiom{Y1}, as no pair $(cd,zz)$ exists in $\Fcl(R)$ for
$c,d,z \in X$ with $cd \neq zz$. Additionally, since $(ab,xy)\notin \Fcl(R)$ and $(ay,ay) \in \tc(R)
\cap \Fcl(R)$, the pair $(R,F)$ also satisfies \axiom{Y2}. Hence, $(R,F)$ is \AF-realizable by
Theorem~\ref{thm:characterization_AF_realized}. Observe further that $F = \FR$ and, thus, $(R,F)$ is
\AF-realized by its canonical DAG $\GG_{R,F}$,   
\end{example}

\begin{figure}
    \centering
    \includegraphics[width=0.7\textwidth]{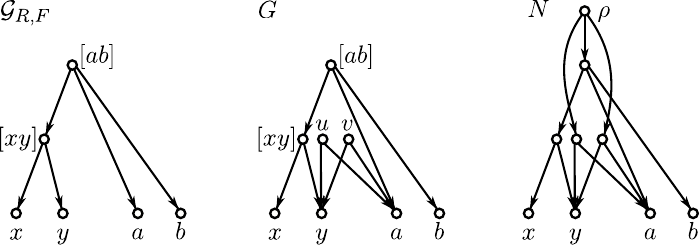}
    \caption{ 
    Shown is the canonical DAG $\GG_{R,F}$ of the two relations $R= \{(xy,ab)\}$ and $F=\{(xy,ay)\}$
    on $\pairs(X)$ with $X = \{a,b,x,y\}$ and its $\FR$-extension $G$. While $\GG_{R,F}$
    \AF-realizes $(R,\FR)=(R,\emptyset)$, it does not \AF-realize $(R,F)$ which, in turn, is
    \AF-realized by $G$, see Example~\ref{exmpl:FR-extension} for more details. In addition, the
    network $N$ obtained from $G$ is shown. 
    }
    \label{fig:FR-extension}
\end{figure}

The following example illustrates that the $\FR$-extension of $\GG_{R,F}$ cannot, in general, be
omitted when constructing a DAG that \AF-realizes $(R,F)$. 

\begin{example}[$\FR$-extension]\label{exmpl:FR-extension}
    Consider the two DAGs $\GG\coloneqq \GG_{R,F}$ and $G$ as shown in
    Figure~\ref{fig:FR-extension}. Here, $\GG$ is the canonical DAG of the two relations $R=
    \{(xy,ab)\}$ and $F=\{(xy,ay)\}$ on $\pairs(X)$ with $X = \{a,b,x,y\}$. This DAG realizes $R$
    and \AF-realizes $(R,\FR)=(R,\emptyset)$. However, $\GG$ does not \AF-realize $(R,F)$ since
    $(xy,ay)\in F$ but $\lca_{\GG}(xy) \prec_{\GG} \lca_{\GG} (ay)$. The DAG $G$ is the
    $\FR$-extension of $\GG$ which is obtained from $\GG$ by a single $ay$-extension. This ensures
    that $\lca_{G} (ay)$ is not well-defined and, therefore, that \axiom{F} is satisfied by $G$ and
    $F$. Hence, in accordance with Theorem~\ref{thm:characterization_AF_realized}, $G$ \AF-realizes
    $(R,F)$. A phylogenetic network that \AF-realizes $(R,F)$ can be obtained from $G$ according to
    Lemma~\ref{lem:DAG2Network}, see the network $N$ in Figure~\ref{fig:FR-extension}.
\end{example}

We now argue that it can be decided in polynomial time if a pair of relations is \AF-realizable and,
in the affirmative case, a network \AF-realizing the pair can be constructed.
Algorithm~\ref{alg:AF-REAL} follows the characterization directly: it computes the closure,
checks the two conditions \axiom{Y1} and \axiom{Y2}, and then builds the canonical DAG and its 
$\FR$-extension together with a canonical network.

\begin{figure}
\centering
\begin{minipage}{0.9\textwidth}
\begin{algorithm}[H] 
\caption{\textsc{[Strict] \AF-Realizability}}
  \begin{algorithmic}[1]
  \Require  A pair $(R,F)$ of relations on $\pairs(X)$
    \Ensure The $\FR$-extension $G$ of the canonical DAG $\GG_{R,F}$ and a phylogenetic network $N$ [strictly] \AF-realizing $(R,F)$ if $(R,F)$ is 
    [strictly] \AF-realizable and, otherwise, \texttt{false} is returned
    \State Compute $\tc(R)$
    \State Compute $\Fcl(R)$ according to Theorem~\ref{thm:polynomial_construction_Fclosure} \label{l:R+}
	 \If{$(R,F)$ satisfies Condition \axiom{Y1} and \axiom{Y2} [and $\tc(R)$ is asymmetric]}   \label{l:Y12}
     \State Compute the canonical DAG $\GG_{R,F}$ as in Definition~\ref{def:canonDAGN}
     \State Compute the $\FR$-extension $G$ of $\GG_{R,F}$ 
     \State Compute the network $N$ obtained from $G$ according to Lemma~\ref{lem:DAG2Network}
	  \State \Return $G$ and  $N$ 
 	 \Else \ 	\Return \texttt{false} \EndIf
  \end{algorithmic}
  \label{alg:AF-REAL}
\end{algorithm}		
\end{minipage}
\end{figure}

\begin{theorem}\label{thm:alg:AF-REAL}
For a pair $(R,F)$ of relations on $\pairs(X)$, verifying whether $(R,F)$ is \AF-realizable and, if
so, constructing a DAG or phylogenetic network that \AF-realizes $(R,F)$ can be done in polynomial
time in $|X|$ using Algorithm~\ref{alg:AF-REAL} without application of ``[and $\tc(R)$ is
asymmetric]" in Line~\ref{l:Y12}. In particular, the canonical DAG $\GG_{R,F}$ can be constructed in
polynomial time. 
\end{theorem}
\begin{proof}
	Let $(R,F)$ be a pair of relations on $\pairs(X)$. We briefly discuss the correctness of
	Algorithm~\ref{alg:AF-REAL}. Observe first that the $\FR$-extension $G$ of $\GG_{R,F}$ and the
	network $N$ is returned if and only if $(R,F)$ satisfies Condition \axiom{Y1} and \axiom{Y2} which
	is, by Theorem~\ref{thm:characterization_AF_realized}, precisely if $(R,F)$ is \AF-realizable. In
	particular, Theorem~\ref{thm:characterization_AF_realized} implies that, in this case, $G$ and $N$
	\AF-realize $(R,F)$ and that $N$ is phylogenetic.

	We discuss now the runtime of Algorithm~\ref{alg:AF-REAL}. To this end, observe that the
	transitive closure $\tc(R)$ of $R$ can be computed in $O(|\support_R|^ 3)$-time using e.g.\ the
	Floyd-Warshall Algorithm \cite{cormen2022introduction}. Since $\support_R \subseteq \pairs(X)$ and
	$|\pairs(X)| = \binom{|X|}{2} + |X| \in O(|X|^2)$, it follows that $\tc(R)$ can be computed in
	polynomial time in $|X|$. By Theorem~\ref{thm:polynomial_construction_Fclosure}, $\Fcl(R)$ can be
	computed in polynomial time in $|X|$. In particular, as outlined at the end of the proof of
	Theorem~\ref{thm:polynomial_construction_Fclosure}, $|\Fcl(R)|\in O(|X|^4)$. We now analyze the
	time complexity to check \axiom{Y1} and \axiom{Y2}. For \axiom{Y1}, we simply scan through all
	$O(|X|^4)$ elements of $\Fcl(R)$ and check that it is not of the form $(ab,xx)$ with $ab\neq xx$.
	For \axiom{Y2}, we check for all $(ab,cd) \in R$ with $(cd,ab) \notin \tc(R)$ that it holds that
	$(cd,ab) \notin \Fcl(R)$. Since we can compute the transitive closure in polynomial time in $|X|$
	and since $|R|, |\tc(R)|,|\Fcl(R)|\in O(|X|^4)$, \axiom{Y2} can be checked in polynomial time in
	$|X|$. To compute the canonical DAG $\GG_{R,F}$, we must first determine $\sim_{\Fcl(R)}$, the set
	$Q$ of $\sim_{\Fcl(R)}$-classes, and $\le_{\Fcl(R)}$ which can all be determined in polynomial
	time in $|X|$, since their computation requires only a finite number of comparison of elements in
	$\Fcl(R)$. It is now straightforward to verify that the canonical DAG $\GG_{R,F}$ can be computed
	in polynomial time in $|X|$. To obtain the $\FR$-extension $G$ of $\GG_{R,F}$, we must first
	determine the sets $\support_F$ and $\support^+_R$. The set $\support^+_R$ has already been
	determined when computing $\Fcl(R)$. The set $\support_F$ can be determined by simply scanning all
	$O(|X|^4)$ elements in $F$. Within the same time complexity, $\support_F \setminus \support^+_R$
	can be computed. Now we apply for all $xy \in \support_F \setminus \support^+_R$ an
	$xy$-extension. In each such $xy$-extension two new vertices $u,v$ and the arcs
	$\{(u,x),(u,y),(v,x),(v,y)\}$ are added; a task that can be done in constant time for each single
	$xy$-extension. Since $\support_F\in O(|X|^4)$, it follows that the $\FR$-extension $G$ of
	$\GG_{R,F}$ can be computed in polynomial time in $|X|$. Furthermore, it is an easy task to verify
	that the network $N$ obtained from $G$ according to Lemma~\ref{lem:DAG2Network} can be constructed
	in polynomial time in $|X|$. In summary, Algorithm~\ref{alg:AF-REAL} runs in polynomial time in
	$|X|$.
\end{proof}

\subsection{Characterization of Strict \AF-Realizability}

In analogy to the strengthening of realization to strict realization, we now 
introduce strict
\AF-realization. The forbidden part remains unchanged, while strictness concerns only the
realization of the required relation $R$. We first show that strict \AF-realizability is a special
case of \AF-realizability and then use this connection to characterize strict \AF-realizability in 
the remainder of this section.

\begin{definition}[Strict \AF-Realization]
A pair $(R, F)$ of relations on $\pairs(X)$ is \emph{strict \AF-realizable} if there is a DAG $G$ on
$X$ such that $R$ is strictly realized by $G$ and \axiom{F} holds for $G$ and $F$. In this case, we
say that $(R,F)$ \emph{is strictly \AF-realized by} $G$. 
\end{definition}

Similar to Lemma~\ref{lem:3.3-lamsh}, we show that strict \AF-realizability is just a special case
of \AF-realizability. 

\begin{lemma} \label{lemma:AFstrictly_iff_AF_asymmetric}
Let $(R,F)$ be a pair of relations on $\pairs(X)$. Then $(R,F)$ is strictly \AF-realized by $G$ if
and only if $(R,F)$ is \AF-realized by $G$ and $\tc(R)$ is asymmetric. 
\end{lemma}
\begin{proof}
	Let $(R,F)$ be a pair of relations on $\pairs(X)$ which is strictly \AF-realized by $G$. By
	definition, $R$ is strictly realized by $G$. Lemma~\ref{lem:3.3-lamsh} implies that $\tc(R)$ is
	asymmetric and that $R$ is realized by $G$. Moreover, \axiom{F} is satisfied by assumption and,
	thus, $(R,F)$ is \AF-realized by $G$.

	Suppose now that $G$ \AF-realizes $(R,F)$ and that $\tc(R)$ is asymmetric. By assumption,
	\axiom{F} must hold. By definition, $G$ realizes $R$. Hence, we can apply
	Lemma~\ref{lem:3.3-lamsh} to conclude that $R$ is strictly realized by $G$. In summary, $(R,F)$ is
	strictly \AF-realized by $G$.
\end{proof}

The latter result also allows us to generalize \cite[Thm~33]{LAMSH:25}
and to characterize strict \AF-realizability as follows.

\begin{theorem}\label{thm:character_strictly_AF}
For a pair $(R,F)$ of relations on $\pairs(X)$ the following statements are equivalent: 
\begin{enumerate}[label={(\arabic*)}]
    \item $(R,F)$ is strictly \AF-realizable. 
    \item $(R,F)$ satisfies \axiom{Y1} and \axiom{Y2} and $\tc(R)$ is asymmetric. 
    \item $(R,F)$ is \AF-realized by the $\FR$-extension of its canonical DAG $\GG_{R,F}$ and $\tc(R)$ is asymmetric.
    \item $(R,F)$ is strictly \AF-realized by the $\FR$-extension of its canonical DAG $\GG_{R,F}$. 
    \item $(R,F)$ is strictly \AF-realized by the phylogenetic network $N$ obtained from the $\FR$-extension of $\GG_{R,F}$ according to Lemma~\ref{lem:DAG2Network}.
\end{enumerate}
\end{theorem}
\begin{proof}
	Let $(R,F)$ be a pair of relations on $\pairs(X)$. By
	Lemma~\ref{lemma:AFstrictly_iff_AF_asymmetric}, $(R,F)$ is strictly \AF-realizable if and only if
	$(R,F)$ is \AF-realizable and $\tc(R)$ is asymmetric. The latter holds if and only if
	Statement~(2) or equivalently Statement (3) holds (cf.\
	Theorem~\ref{thm:characterization_AF_realized}). In addition, by
	Lemma~\ref{lemma:AFstrictly_iff_AF_asymmetric}, Statement (3) and (4) are equivalent. To summarize
	so far, Statement (1) to (4) are equivalent. Suppose now that Statement~(3) holds. This together
	with Theorem~\ref{thm:characterization_AF_realized} implies that $(R,F)$ is \AF-realized by the
	phylogenetic network $N$ obtained from the $\FR$-extension of $\GG_{R,F}$ according to
	Lemma~\ref{lem:DAG2Network} and that $\tc(R)$ is asymmetric. Then,
	Lemma~\ref{lemma:AFstrictly_iff_AF_asymmetric} implies Statement (5). Conversely, Statement (5)
	implies Statement~(1) which completes this proof.
\end{proof}

Theorem~\ref{thm:character_strictly_AF} yields an immediate algorithmic consequence. Indeed, strict
\AF-realizability differs from \AF-realizability only by the additional requirement that
$\tc(R)$ be asymmetric. Since this condition can be checked in polynomial time, strict
\AF-realizability can likewise be decided in polynomial time. Moreover, whenever $(R,F)$ is
strictly \AF-realizable, a DAG or phylogenetic network that strictly \AF-realizes $(R,F)$ can be
constructed within the same asymptotic time bound.

\begin{theorem}\label{thm:alg:strictAF-REAL}
For a pair $(R,F)$ of relations on $\pairs(X)$, verifying whether $(R,F)$ is strictly \AF-realizable
and, if so, constructing a DAG or phylogenetic network that strictly \AF-realizes $(R,F)$ can be
done in polynomial time in $|X|$ using Algorithm~\ref{alg:AF-REAL} with application of ``[and
$\tc(R)$ is asymmetric]" in Line~\ref{l:Y12}.
\end{theorem}
\begin{proof}
	The correctness of Algorithm~\ref{alg:AF-REAL} with application of ``[and $\tc(R)$ is asymmetric]"
	in Line~\ref{l:Y12} is a direct consequence of Theorem~\ref{thm:character_strictly_AF} and similar
	arguments as used in Theorem~\ref{thm:alg:AF-REAL}. The only overhead in terms of runtime compared
	to Algorithm~\ref{alg:AF-REAL} without application of ``[and $\tc(R)$ is asymmetric]" in
	Line~\ref{l:Y12} is checking whether $\tc(R)$ is asymmetric. It is easy to verify that this task
	can be performed in polynomial time in $|X|$. This together with Theorem~\ref{thm:alg:AF-REAL}
	implies the polynomial runtime. 
\end{proof}
%

%
%

\section{Alternative Definitions of Forbidden LCA-constraints}
\label{sec:alternative}

The interested reader may have wondered why we impose the condition
\begin{itemize}
    \item[]\axiom{F}\ \ \ If $(ab,cd) \in F$ and $\lca_G(ab)$ and $\lca_G(cd)$ are well-defined, then $\lca_G(ab) \nprec_G \lca_G(cd)$
\end{itemize}
instead of the following natural variant
\begin{itemize}
    \item[]\axiom{F$^\npreceq$}\ \  If $(ab,cd) \in F$ and $\lca_G(ab)$ and $\lca_G(cd)$ are well-defined, then $\lca_G(ab) \npreceq_G \lca_G(cd)$
\end{itemize}
or 
\begin{itemize}
    \item[]\axiom{F$^{\mathbf{lca}}$}\ \  If $(ab,cd) \in F$, then $\lca_G(ab)$ and $\lca_G(cd)$ are well-defined and $\lca_G(ab) \nprec_G \lca_G(cd)$. 
\end{itemize}

The latter variant is particularly relevant, since it rules out the convenient escape route used in
Section~\ref{sec:AFLCA}, where an LCA occurring in a forbidden constraint is made non-unique
whenever possible. 
Thus, in the spirit of completeness and to reassure any reader still mildly suspicious of our 
previous choice, we also provide

\begin{definition}[\AFeq-Realization and \AFlca-Realization]
	A pair $(R, F)$ of relations on $\pairs(X)$ is \emph{\AFeq-realizable} (resp.\
	\emph{\AFlca-realizable}) if there is a DAG $G$ on $X$ such that $R$ is realized by $G$ and
	\axiom{F$^\npreceq$} (resp.\ \axiom{F$^{\mathbf{lca}}$}) holds for $G$ and $F$. In this case, we say
	that $(R,F)$ is \emph{\AFeq-realized} (resp.\ \emph{\AFlca-realized}) by $G$.  
\end{definition}

In the following, we characterize realizability under these two alternative definitions and, in
each affirmative case, provide an explicit construction of a DAG that \AFeq-realizes or,
respectively, \AFlca-realizes $(R,F)$. We conclude the section by showing that both notions are
proper special cases of \AF-realizability.

We begin with \AFeq-realizability, whose characterization is considerably simpler than that of
\AF-realizability. Clearly, if $(R,F)$ is \AFeq-realizable, then $R$ must be realizable. In this
case, the closure $\cl_{\emptyset}(R)$ contains precisely the LCA-constraints that are satisfied by
every DAG realizing $R$. In particular, for every $(ab,cd)\in\cl_{\emptyset}(R)$ and every DAG
$G$ realizing $R$, we have
\[
\lca_G(ab)\preceq_G\lca_G(cd);
\]
see Lemma~20 in \cite{LAMSH:25}. Consequently, if
$(ab,cd)\in F\cap\cl_{\emptyset}(R)$, then membership in
$\cl_{\emptyset}(R)$ requires
$\lca_G(ab)\preceq_G\lca_G(cd)$, whereas \axiom{F$^\npreceq$} requires
$\lca_G(ab)\npreceq_G\lca_G(cd)$. Hence,
$F\cap\cl_{\emptyset}(R)=\emptyset$ is necessary. As the following theorem shows, this condition
together with the realizability of $R$ is also sufficient.

\begin{theorem}\label{thm:AFeq_realizable}
A pair $(R,F)$ of relations is \AFeq-realizable if and only if $R$ is realizable and $F\cap
\cl_{\emptyset}(R)=\emptyset$. In this case, the $\FR$-extension of the canonical DAG
$\GG_{R,\emptyset}$ \AFeq-realizes $(R,F)$.
\end{theorem}
\begin{proof}
	We make frequent use of the fact that $\cl_{\emptyset}(R)=\cl(R)$, cf.\
	Observation~\ref{obs:cl_emptyset=cl}. Suppose that $(R,F)$ is \AFeq-realized by the DAG $G$. By
	definition, $R$ is realized by $G$. Assume, for contradiction, that there is some $(ab,xy)\in
	F\cap \cl_{\emptyset}(R)$. Then, \cite[L~20]{LAMSH:25} and $(ab,xy) \in \cl(R)$ implies that
	$\lca_G(ab)$ and $\lca_G(xy)$ are well-defined and $\lca_G(ab)\preceq_G \lca_G(xy)$. But then $G$
	and $F$ do not satisfy \axiom{F$^\npreceq$}; a contradiction to the definition of
	\AFeq-realizability. 

	Assume now that $R$ is realizable and $F\cap \cl_{\emptyset}(R)=\emptyset$. By
	\cite[L~43]{LAMSH:25}, there is a DAG $G$ that realizes $R$ and satisfies $\cl_\emptyset(R) =
	\rel_G$. Assume, for contradiction, that \axiom{F$^\npreceq$} is violated and thus that, for some
	$(ab,xy)\in F$, it holds that $\lca_G(ab)$ and $\lca_G(xy)$ are well-defined and
	$\lca_G(ab)\preceq_G \lca_G(xy)$. By definition of $\rel_G$, it holds that $(ab,xy)\in
	\rel_G=\cl_{\emptyset}(R)$; a contradiction to $F\cap \cl_{\emptyset}(R)=\emptyset$. 

	Suppose now that $(R,F)$ is \AFeq-realizable. By definition, $R$ is realizable. Hence,
	$(R,\emptyset)$ is \AF-realizable, since \axiom{F} is vacuously true. By
	Theorem~\ref{thm:characterization_AF_realized}, $\GG \coloneqq \GG_{R,\emptyset}$ \AF-realizes
	$(R,\emptyset)$ and, thus, $\GG$ realizes in particular $R$. Now, let $(ab,xy) \in \FR$ and, thus,
	$ab, xy \in \support^+_R$. This implies that the vertices $\lca_\GG(ab)$ and $\lca_\GG(xy)$ are
	well-defined. Suppose, for contradiction, that $\lca_\GG(ab) \preceq_\GG \lca_\GG(xy)$. Note that
	$\GG$ coincides with the canonical DAG $\GG_R$ as defined in \cite{LAMSH:25}. This allows us to
	apply \cite[L~41]{LAMSH:25}, which implies that $(ab,xy) \in \cl_\emptyset(R) = \cl(R)$. Since
	$\FR\subseteq F$ it follows that $(ab,xy) \in F \cap \cl_\emptyset(R)$. As shown in the previous
	paragraphs, this implies that $(R,F)$ is not \AFeq-realizable; a contradiction to the assumption.
	Hence, $\lca_\GG(ab) \not\preceq_\GG \lca_\GG(xy)$ must hold for all $(ab,xy) \in \FR$. Therefore,
	\axiom{F$^\npreceq$} holds for $\FR$ and $\GG$. This together with the fact that $\GG$ realizes
	$R$ implies that $\GG$ \AFeq-realizes $(R,\FR)$. We can now re-use similar arguments as used in
	the proof of Proposition~\ref{prop:RF_real_iff_RFR_real}(2) to conclude that the $\FR$-extension
	of $\GG$ \AFeq-realizes $(R,F)$.
\end{proof}

We continue by characterizing \AFlca-realizable pairs $(R,F)$ of relations. In contrast to
\axiom{F} and \axiom{F$^\npreceq$}, Condition~\axiom{F$^{\mathbf{lca}}$} also requires the LCAs
$\lca_G(ab)$ to be well-defined for all $ab\in\support_F$. To enforce this requirement, we consider
the extension $R^{\lca}$ of $R$ obtained by adding the constraints $(ab,ab)$ for all
$ab\in\support_F$. Thus,
\[
R^{\lca}\coloneqq R\cup\{(ab,ab)\mid ab\in\support_F\}.
\]
It follows that
$
\support^+_{R^{\lca}}=\support^+_R\cup\support_F$,
and, in every DAG $G$ realizing $R^{\lca}$, the LCA $\lca_G(ab)$ is well-defined for each
$ab\in\support_F$. This allows us to characterize \AFlca-realizability in terms of
\AF-realizability applied to the extended pair $(R^{\lca},F)$.

\begin{theorem}\label{thm:characterization_AFlca_realized}
Let $(R,F)$ be a pair of relations on $\pairs(X)$ and let $R^{\lca} \coloneqq R \cup \{(ab,ab) \mid
ab \in \support_F\}$ be a relation on $\pairs(X)$. Then a DAG $G$ \AFlca-realizes $(R,F)$ if and only if $G$ \AF-realizes
$(R^{\lca},F)$.
\end{theorem}
\begin{proof}
	Let $(R,F)$ be a pair of relations on $\pairs(X)$ and let $R^{\lca} \coloneqq R \cup \{(ab,ab)
	\mid ab \in \support_F\}$. It is straightforward to verify that $\tc(R^{\lca}) = \tc(R) \cup
	\{(ab,ab) \mid ab \in \support_F\}$. Suppose that $(R,F)$ is \AFlca-realized by the DAG $G$. By
	definition, $G$ and $F$ satisfy \axiom{F$^{\mathbf{lca}}$} and, thus, $G$ and $F$ satisfy
	\axiom{F}. Moreover, by definition, $G$ realizes $R$. To prove that $G$ realizes $R^{\lca}$, let
	$(p,q) \in R^{\lca}$. We distinguish now between the two cases $(p,q) \in R$ and $(p,q) \notin R$.
	Suppose first that $(p,q) \in R$. Since $G$ realizes $R$, the vertices $\lca_G(p)$ and $\lca_G(q)$
	are well-defined and one of the Conditions \axiom{I1} and \axiom{I2} is satisfied by $G$ and $R$.
	Suppose now that $(q,p) \notin \tc(R^{\lca})$. Then $\tc(R) \subseteq \tc(R^{\lca})$ implies that
	$(q,p) \notin \tc(R)$. Since $G$ and $R$ satisfy \axiom{I1}, we have $\lca_G(p)\prec_G\lca_G(q)$
	and it follows that $G$ and $R^{\lca}$ satisfy \axiom{I1}. Suppose now that $(q,p) \in
	\tc(R^{\lca}) = \tc(R) \cup \{(ab,ab) \mid ab \in \support_F\}$. Assume, for contradiction, that
	$(q,p) \notin \tc(R)$. Then, $(q,p) \in \{(ab,ab) \mid ab \in \support_F\}$ and $q = p$ holds.
	Since $(p,q) \in R$, we obtain $(q,p) \in R \subseteq \tc(R)$; a contradiction to the assumption.
	Hence, $(q,p) \in \tc(R)$ holds. Since $G$ and $R$ satisfy \axiom{I2}, we have $\lca_G(p)
	=\lca_G(q)$ and it follows that $G$ and $R^{\lca}$ satisfy \axiom{I2}. Now, let $(p,q) \in
	R^{\lca} \setminus R$. Hence, $(p,q)$ is of the form $(ab,ab)$ for some $ab \in \support_F$. Thus,
	$(q,p) \in R^{\lca} \subseteq \tc(R^{\lca})$ holds. In particular, since $G$ \AFlca-realizes
	$(R,F)$ and $ab \in \support_F$, the vertex $\lca_G(ab)$ is well-defined. Clearly, $\lca_G(ab) =
	\lca_G(ab)$ holds and $G$ and $R^{\lca}$ satisfy \axiom{I2}. In summary, $G$ \AF-realizes
	$(R^{\lca},F)$.

	Now suppose that $(R^{\lca},F)$ is \AF-realized by a DAG $G$. To show that $G$ realizes $R$, let
	$(p,q) \in R$. Suppose that $(q,p) \notin \tc(R)$. We first argue that $(q,p) \notin
	\tc(R^{\lca})$. Assume, for contradiction, that $(q,p) \in \tc(R^{\lca})$. Since $\tc(R^{\lca}) =
	\tc(R) \cup \{(ab,ab) \mid ab \in \support_F\}$ and $(q,p) \notin \tc(R)$, we obtain $(q,p) \in
	\{(ab,ab) \mid ab \in \support_F\}$ and, therefore, $p = q$. This and $(p,q) \in R$ implies $(q,p)
	\in R \subseteq \tc(R)$; a contradiction to the assumption. Hence, $(q,p) \notin \tc(R^{\lca})$
	must hold. Since $G$ and $R^{\lca}$ satisfy \axiom{I1}, we have $\lca_G(p)\prec_G\lca_G(q)$ and it
	follows that $G$ and $R$ also satisfy \axiom{I1}. Suppose now that $(q,p) \in \tc(R) \subseteq
	\tc(R^{\lca})$. Since $G$ and $R^{\lca}$ satisfy \axiom{I2}, $\lca_G(p)=\lca_G(q)$ holds and it
	follows that $G$ and $R$ also satisfy \axiom{I2}. Therefore, $G$ realizes $R$. Moreover, since $G$
	\AF-realizes $(R^{\lca},F)$, the vertex $\lca_G(ab)$ is well-defined for all $ab \in \support_F$.
	This together with the fact that $G$ and $F$ satisfy \axiom{F} implies that $G$ and $F$ satisfy
	\axiom{F$^{\mathbf{lca}}$}. In summary, $G$ \AFlca-realizes $(R,F)$.
\end{proof}

As a direct consequence of Theorems~\ref{thm:characterization_AF_realized}, \ref{thm:alg:AF-REAL},
\ref{thm:AFeq_realizable}, and \ref{thm:characterization_AFlca_realized}, we obtain

\begin{corollary}
For a pair $(R,F)$ of relations on $\pairs(X)$, verifying whether $(R,F)$ is \AFeq-realizable,
resp., \AFlca-realizable, and, if so, constructing a DAG or phylogenetic network that
\AFeq-realizes, resp., \AFlca-realizes $(R,F)$ can be done in polynomial time in $|X|$.
\end{corollary}

Additionally, it is easy to verify that if $G$ and $F$ satisfy \axiom{F$^\npreceq$} or
\axiom{F$^{\mathbf{lca}}$} than they satisfy \axiom{F}. In other words, the class of
\AFeq-realizable, resp., \AFlca-realizable pairs of relations forms a subclass of the class of
\AF-realizable pairs of relations. Moreover, one can find \AF-realizable pairs that are not
\AFeq-realizable or \AFlca-realizable as shown in the following 

\begin{example}[\AFeq-realizability and \AFlca-realizability are \emph{proper} special cases of \AF-realizability] \label{exmpl:AFlca}
First, consider the pair $(R,F)$ of relations $R = \{(xy,xz),(yy,yz)\}$ and $F=\{(yz,xz)\}$ as in
Example~\ref{exmpl:Fcsym}, see also Figure~\ref{fig:GrealR-notRF.jpg}. In this example, $(R,F)$ is
\AF-realizable, but since $(yz,xz) \in F \cap \cl_{\emptyset}(R)$ and by
Theorem~\ref{thm:AFeq_realizable}, $(R,F)$ is not \AFeq-realizable. In particular, since
$(xz,yz),(yz,xz) \in \Fcl(R)$ and $F =\FR$, Lemma~\ref{lemma:Fclosure_preceq} implies that
$\lca_G(xz)=\lca_G(yz)$ holds for any DAG $G$ that \AF-realizes $(R,F)$ and, thus, also for all DAGs
$G$ that \AFeq-realizes $(R,F)$. This, in turn, implies that \axiom{F$^\npreceq$} can never be
satisfied for $(R,F)$. In summary, \AFeq-realizability is a proper special case of
\AF-realizability.

Now, consider the relations $R= \{(xy,ab)\}$ and $F=\{(xy,ay),(ay,ab)\}$ on $\pairs(X)$ with
$X=\{a,b,x,y\}$ as illustrated in Figure~\ref{fig:AFlca}. Since $(R,F)$ satisfies \axiom{Y1} and
\axiom{Y2}, Theorem~\ref{thm:characterization_AF_realized} implies that the $\FR$-extension $G$ of
the canonical DAG $\GG_{R,F}$ \AF-realizes $(R,F)$. Now consider $R^{\lca} = \{(xy,ab),
(ay,ay)\}$. Then, $(R^{\lca},F)$ does not satisfy \axiom{Y2}, since $(xy,ab) \in
R^{\lca}$ and $(ab,xy) \notin \tc(R^{\lca})$, but $(ab,xy) \in \Fcl(R^{\lca})$. Hence, by
Theorem~\ref{thm:characterization_AF_realized}, $(R^{\lca},F)$ is not \AF-realizable and
Theorem~\ref{thm:characterization_AFlca_realized} implies that $(R,F)$ is not \AFlca-realizable.
Therefore, also \AFlca-realizability is a proper special case of \AF-realizability.
\end{example}

\begin{figure}
    \centering
    \includegraphics[width=0.8\linewidth]{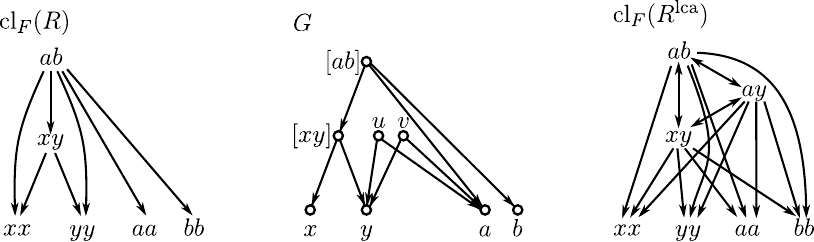}
    \caption{
    Shown are $\Fcl(R)$, $\Fcl(R^{\lca})$, and the $\FR$-extension $G$ of $\GG_{R,F}$
    for the relations  $R= \{(xy,ab)\}$,  
    $F=\{(xy,ay),(ay,ab)\}$ and $R^{\lca} = \{(xy,ab), (xy,xy), (ab,ab), (ay,ay)\}$ on $\pairs(X)$ with $X = \{a,b,x,y\}$.
    Here we draw an arc $p\to q$ precisely if $(q,p) \in \Fcl$ but 
    omitted arcs of the form $(p,p)$, $\Fcl\in \{\Fcl(R), \Fcl(R^{\lca})\}$.
    While $G$ \AF-realizes $(R,F)$, the pair 
    $(R,F)$ is not \AFlca-realizable, 
    see Example~\ref{exmpl:AFlca} for more details.
    }
    \label{fig:AFlca}
\end{figure}

%
%
\section{Classical Closure} 
\label{sec:classical-closure}

In the preceding sections, we introduced the closure $\Fcl(R)$ of a relation $R$ with respect to a
relation $F$ and used it to construct the canonical DAG associated with an \AF-realizable pair of relations.
For such \AF-realizable pairs, however, there is another, arguably more intrinsic, way to define such a
closure. This alternative mirrors the standard notion of closure for triplets and quartets on trees
\cite{SH:18,MaayanLevy24,Bryant1997,BS:95}. More precisely, let $\FCL(R)$ be the intersection of all relations
\[\rel_G = \{(ab,xy) \mid \lca_G(ab), \lca_G(xy) \text{ are well-defined and } \lca_G(ab)\preceq_G
\lca_G(xy)\}\] 
induced by DAGs $G$ that \AF-realize $(R,F)$. 
Since $\rel_G$ consists precisely of the LCA-constraints satisfied by $G$,
the relation $\FCL(R)$ comprises exactly those constraints that are
satisfied by every DAG that \AF-realizes $(R,F)$.

The main result of this section shows that the two notions of closure coincide for every
\AF-realizable pair $(R,F)$, that is, $ \FCL(R)=\Fcl(R); $ see Theorem~\ref{thm:closure-classic}.
Thus, the constraints obtained combinatorially from $R$ and $F$ by means of the closure $\Fcl(R)$
are precisely the constraints shared by all \AF-realizations of $(R,F)$. 
We start by giving the formal definition of the alternative closure.

\begin{definition}[Classical Closure]
Let $(R,F)$ be an \AF-realizable pair of relations on $\pairs(X)$, and let $\mathfrak{G}$ be the set
of all DAGs on $X$ that \AF-realize $(R,F)$. Then, we define the following intersection
\[
\FCL(R) \coloneqq \bigcap_{G \in \mathfrak{G}} \rel_G.  
\]
\end{definition}

Note that, by definition, $\FCL(R)$ is a relation on $\pairs(X)$. 
In what follows, for every \AF-realizable pair $(R,F)$, we first establish that
$\Fcl(R)$ is precisely the set of LCA-constraints satisfied by the canonical DAG
$\GG_{R,F}$ and then use these auxiliary results to prove that
$\FCL(R)=\Fcl(R)$.

\begin{lemma} \label{lemma:preceq_implies_in_Fclosure}
Let $(R,F)$ be an \AF-realizable pair of relations on $\pairs(X)$ and let $\GG_{R,F}$ be the
canonical DAG of $(R,F)$. Then, for all $ab,xy \in \support^+_{R}$, $\lca_{\GG_{R,F}}(ab)$ and
$\lca_{\GG_{R,F}}(xy)$ are well-defined and the following equivalence holds:
\begin{equation}\label{eq:equiv}
\lca_{\GG_{R,F}}(ab) \preceq_{\GG_{R,F}} \lca_{\GG_{R,F}}(xy) \iff (ab,xy) \in \Fcl(R).    
\end{equation}
\end{lemma}
\begin{proof}
	Let $(R,F)$ be an \AF-realizable pair of relations on $\pairs(X)$ and $\GG_{R,F}$ be the canonical
	DAG of $(R,F)$. Suppose $ab \in \support^+_{R}$. By
	Theorem~\ref{thm:characterization_AF_realized}, $\GG_{R,F}$ realizes, in particular, $R$ and,
	therefore, $\lca_{\GG_{R,F}}(ab)$ is well-defined.

	Moreover, the \emph{if}-direction of Eq.~\ref{eq:equiv} follows from
	Lemma~\ref{lemma:Fclosure_preceq} and Theorem~\ref{thm:characterization_AF_realized}. For the
	\emph{only-if}-direction of Eq.~\ref{eq:equiv}, let $ab, xy \in \support_{R}^+$ and suppose that
	$v\coloneqq \lca_{\GG_{R,F}}(ab) \preceq_{\GG_{R,F}} \lca_{\GG_{R,F}}(xy) \eqqcolon w$. Let $q =
	[ab]$ and $q'=[xy]$ denote the equivalence classes in the quotient poset $(Q,\le_{\Fcl(R)})$. By
	Proposition~\ref{prop:properties_of_canonical_DAG}, $\lca_{\GG_{R,F}}(ab) = [ab] =q$ if $a \neq b$
	and, otherwise, $\lca_{\GG_{R,F}}(ab) = a$ in which case $q=[aa]$. Similarly,
	$\lca_{\GG_{R,F}}(xy) = [xy] =q'$ if $x \neq y$ and, otherwise, $\lca_{\GG_{R,F}}(xy) = x$ and
	$q'=[xx]$. Since $v \preceq_{\GG_{R,F}} w$, there is a $wv$-path in $\GG_{R,F}$. Since $\GG_{R,F}$
	only differs from $\Hasse(Q,\le_{\Fcl(R)})$ by the leaves $[aa]$ which are relabeled by $a$ and
	since $\leq_{\Fcl(R)}$ is transitive, it is straightforward to verify that $q \leq_{\Fcl(R)} q'$.
	By definition of $\leq_{\Fcl(R)}$, it holds that $(ab,xy) \in \Fcl(R)$.
\end{proof}

Lemma~\ref{lemma:Fclosure_preceq} shows that $\Fcl(R)\subseteq \rel_G$ for all DAGs $G$ that
\AF-realize $(R,F)$. We can go even further and show that for every \AF-realizable pair $(R,F)$ of
relations, there exists a DAG $G$ that \AF-realizes $(R,F)$ and satisfies $\Fcl(R) = \rel_G$.
This result will, in particular, be useful to show that $\FCL(R) = \Fcl(R)$.

\begin{lemma} \label{lemma:existence_of_G_with_rel_=_Fclosure} 
For every \AF-realizable pair $(R,F)$ of relations, there exists a DAG $G$ that \AF-realizes $(R,F)$
and satisfies $\rel_G = \Fcl(R)$. 
\end{lemma}
\begin{proof}
	Let $(R,F)$ be an \AF-realizable pair of relations on $\pairs(X)$. We follow now some of the ideas
	as in the proof of \cite[L~43]{LAMSH:25}.

	By Theorem~\ref{thm:characterization_AF_realized}, $(R,\FR)$ is \AF-realized by the canonical DAG
	$\GG \coloneqq \GG_{R,F}$. Let $G$ be the directed graph obtained from $\GG$ by applying an
	$xy$-extension for every $xy \in \pairs(X) \setminus \support^+_{R}$. Note that, since $\GG$
	realizes $R$, for all $ab \in \support^+_{R}$ the LCA $\lca_{\GG}(ab)$ is well-defined. Note that
	no $ab$-extension was applied to obtain $G$ whenever $ab \in \support^+_{R}$. The latter two
	arguments together with Observation~\ref{obs:xy-extension1} imply that $G$ remains a DAG on $X$
	where $\lca_\GG(ab) = \lca_G(ab)$ is well-defined for all $ab \in \support^+_{R}$. Moreover,
	observe that $V(\GG) \subseteq V(G)$. 

	We now prove that $\rel_G = \Fcl(R)$ and start with showing that $\Fcl(R) \subseteq \rel_G$. Let
	$(ab,xy) \in \Fcl(R)$. By Theorem~\ref{thm:polynomial_construction_Fclosure}, $ab,xy \in
	\support_{\Fcl(R)} = \support^+_R$. Since $\GG$ \AF-realizes $(R,\FR)$, the vertices
	$\lca_\GG(ab)$ and $\lca_\GG(xy)$ are well-defined and by Lemma~\ref{lemma:Fclosure_preceq},
	$\lca_\GG(ab) \preceq_\GG \lca_\GG(xy)$. By the arguments in the first paragraph of this proof and
	Observation~\ref{obs:xy-extension1}, we can conclude that $\lca_\GG(ab) = \lca_G(ab)\preceq_G
	\lca_G(xy) = \lca_\GG(xy)$. Hence, $(ab,xy) \in \rel_G$ and, thus, $\Fcl(R) \subseteq \rel_G$. We
	show now that $\rel_G \subseteq \Fcl(R)$. To this end, let $(ab,xy) \in \rel_G$. By definition of
	$\rel_G$, the LCAs $\lca_G(ab)$ and $\lca_G(xy)$ are well-defined and $\lca_G(ab) \preceq_G
	\lca_G(xy)$ holds. This and Observation~\ref{obs:xy-extension1} implies that no $ab$ and
	$xy$-extension was applied on $\GG$ to obtain $G$ and, consequently, $ab, xy \in \support^+_R$.
	This and the arguments in the first paragraph of this proof imply that $\lca_G(ab) = \lca_\GG(ab)
	\preceq_\GG \lca_\GG(xy) = \lca_G(xy)$. This together with
	Lemma~\ref{lemma:preceq_implies_in_Fclosure} implies $(ab,xy) \in \Fcl(R)$ and hence, $\rel_G
	\subseteq \Fcl(R)$. In summary, $\Fcl(R) = \rel_G$ holds. 
 
	It remains to prove that $G$ \AF-realizes $(R,F)$. We start with showing that $G$ realizes $R$ and,
	thus, that \axiom{I1} and \axiom{I2} is satisfied for $G$ and $R$. Let $(ab,cd) \in R \subseteq
	\Fcl(R) = \rel_G$. Hence, $(ab,cd) \in \rel_G$ which implies that $\lca_G(ab)$ and $\lca_G(cd)$
	are well-defined and $\lca_G(ab) \preceq_G \lca_G(cd)$. First, suppose that $(cd,ab) \notin
	\tc(R)$. Since $(R,F)$ is \AF-realizable, $(R,F)$ satisfies \axiom{Y2} and, therefore, it holds
	that $(cd, ab) \notin \Fcl(R)=\rel_G$. Hence, $(cd,ab) \notin \rel_G$ and $\lca_G(cd) \npreceq_G
	\lca_G(ab)$ follows. This together with $\lca_G(ab) \preceq_G \lca_G(cd)$ implies $\lca_G(ab)
	\prec_G \lca_G(cd)$. Therefore, \axiom{I1} is fulfilled. Now suppose $(cd, ab) \in \tc(R)$. By
	Observation~\ref{obs:Fclosure_transitive_and_reflexive}, $\tc(R)\subseteq \Fcl(R)$ and, thus, $(cd,ab) \in
	\Fcl(R) = \rel_G$. Hence, $\lca_G(cd) \preceq_G \lca_G(ab)$ holds. This and $\lca_G(ab) \preceq_G
	\lca_G(cd)$ implies that $\lca_G(ab) = \lca_G(cd)$ must hold. Hence, \axiom{I2} is satisfied.
	In summary, $G$ realizes $R$.

	It remains to show that $G$ and $F$ satisfy \axiom{F}. Note that, by definition, for all $(ab,cd)
	\in \FR$, it holds that $ab,cd \in \support_{R}^+$ and by the arguments in the first paragraph of
	this proof, $\lca_G(ab) = \lca_\GG(ab) \in V(\GG)$ and $\lca_G(cd) = \lca_\GG(cd) \in V(\GG)$. Now
	assume, for contradiction, that there exists $(ab,cd) \in \FR$ such that $\lca_G(ab) \prec_G
	\lca_G(cd)$. By Observation~\ref{obs:xy-extension1}, $\lca_G(ab) \prec_G \lca_G(cd)$ implies
	$\lca_\GG(ab) \preceq_\GG \lca_\GG(cd)$. However, if $\lca_\GG(ab) = \lca_\GG(cd)$, then
	$\lca_G(ab) = \lca_G(cd)$ would follow. Thus, $\lca_\GG(ab) \prec_\GG \lca_\GG(cd)$ must hold; a
	contradiction to $\GG$ \AF-realizing $(R,\FR)$. In summary, $G$ \AF-realizes $(R,\FR)$. Lastly,
	since we applied an $xy$-extension for all $xy \in \pairs(X) \setminus \support^+_{R}$ and, thus,
	in particular, for all $xy \in \support_F \setminus \support^+_R$ to obtain $G$,
	Observation~\ref{obs:xy-extension1} implies that $\lca_G(xy)$ is not well-defined for all $xy \in
	\support_F \setminus \support^+_R$. Hence, $G$ satisfies \axiom{F} for all $(ab,xy) \in F$. In
	conclusion, $G$ \AF-realizes $(R,F)$.
\end{proof}

We now prove that $\Fcl(R) = \FCL(R)$ and thereby generalize \cite[Thm~44]{LAMSH:25}.

\begin{theorem}\label{thm:closure-classic}
Let $F$ be a relation on $\pairs(X)$ and let $\mathcal D_F\coloneqq \{R \mid R \text{ is a relation
on } \pairs(X) \text{ and } (R,F)\text{ is \AF-realizable}\}. $ Then, for every $R\in\mathcal D_F$,
it holds that $ \FCL(R)=\Fcl(R). $ Moreover, $\FCL\colon\mathcal D_F\to\mathcal D_F$ is a closure
operator and can be computed in polynomial time in $|X|$.
\end{theorem}
\begin{proof}
	Let $(R,F)$ be an \AF-realizable pair of relations on $\pairs(X)$, and let $\mathfrak{G}$ be the
	set of all DAGs on $X$ that \AF-realize $(R,F)$. We first show that $\Fcl(R) \subseteq \FCL(R)$.
	Let $G \in \mathfrak{G}$ and we argue that $\rel_G \in \mathfrak{R}_{R,F}$. By Proposition~21 in
	\cite{LAMSH:25}, $\rel_G$ is $\support^+_{\rel_G}$-reflexive, transitive, and cross-consistent.
	Moreover, since $G$ realizes $R$, Lemma~7 in \cite{LAMSH:25} implies that $R \subseteq \rel_G$
	and, therefore, $\support_R^+\subseteq\support_{\rel_G}^+$. Hence, $\rel_G$ is
	$\support_R^+$-reflexive. It remains to argue that $\rel_G$ is $F$-csym. Let $(p,q) \in F \cap
	\rel_G$. Since $(p,q) \in \rel_G$, the LCAs $\lca_G(p)$ and $\lca_G(q)$ are well-defined and
	satisfy $\lca_G(p) \preceq_G \lca_G(q)$. This together with $(p,q) \in F$ and the fact that $G$
	\AF-realizes $(R,F)$ implies that $\lca_G(p) \nprec_G \lca_G(q)$. Hence, $\lca_G(p) = \lca_G(q)$
	and, therefore, $(q,p) \in \rel_G$. Since the pair $(p,q) \in F \cap \rel_G$ was chosen
	arbitrarily, it follows that $\rel_G$ is $F$-csym. In summary, $\rel_G \in \mathfrak{R}_{R,F}$.
	Consequently, $\Fcl(R) = \bigcap_{R' \in \mathfrak{R}_{R,F}} R' \subseteq \rel_G$ for every $G \in
	\mathfrak{G}$. Hence, $\Fcl(R) \subseteq \bigcap_{G \in \mathfrak{G}} \rel_G = \FCL(R)$. We now
	show that $\FCL(R) \subseteq \Fcl(R)$. Since $(R,F)$ is \AF-realizable,
	Lemma~\ref{lemma:existence_of_G_with_rel_=_Fclosure} implies that there exists a DAG $G$ that
	\AF-realizes $(R,F)$ and satisfies $\rel_G = \Fcl(R)$. In particular, $G \in \mathfrak{G}$. This
	and $\FCL(R) = \bigcap_{G' \in \mathfrak{G}} \rel_{G'}$ implies that $\FCL(R) \subseteq \rel_G =
	\Fcl(R)$. In summary, $\FCL(R) = \Fcl(R)$. By Theorem~\ref{thm:polynomial_construction_Fclosure},
	$\FCL(R) = \Fcl(R)$ can be constructed in polynomial time in $|X|$ whenever $(R,F)$ is
	\AF-realizable.

  Let $F$ be a relation on $\pairs(X)$.
	It remains to show that $\FCL\colon\mathcal{D}_F\to\mathcal{D}_F$ is a closure operator. First, we
	verify that $\FCL$ maps $\mathcal{D}_F$ into itself. Let $R\in\mathcal{D}_F$. By
	Lemma~\ref{lemma:existence_of_G_with_rel_=_Fclosure}, there exists a DAG $G$ that \AF-realizes
	$(R,F)$ and satisfies $ \FCL(R)=\Fcl(R)=\rel_G$. By \cite[L~6]{LAMSH:25}, $G$ realizes $\rel_G$
	and it follows that
	$G$ also \AF-realizes $(\rel_G,F)=(\FCL(R),F)$. Hence, $\FCL(R)\in\mathcal{D}_F$. Extensivity
	and monotonicity of $\FCL$ on $\mathcal{D}_F$ now follow from
	Proposition~\ref{prop:closure_operator} and the equality $\FCL(R)=\Fcl(R)$ for every
	$R\in\mathcal{D}_F$. Finally, since $\FCL(R)\in\mathcal{D}_F$, we obtain $ \FCL(\FCL(R)) =
	\Fcl(\Fcl(R)) = \Fcl(R) = \FCL(R)$.  Thus, $\FCL$ is idempotent and, therefore,
	$\FCL\colon\mathcal{D}_F\to\mathcal{D}_F$ is a closure operator.
\end{proof}

%
%
\section{Summary and Outlook}
\label{sec:summary}

In this contribution, we introduced and studied three variants of the realization problem for pairs
$(R,F)$ of required and forbidden LCA-constraints, namely, \AF-realizability,
\AFeq-realizability, and \AFlca-realizability. For each of these variants, we obtained a
characterization of the pairs $(R,F)$ that admits a realization and derived a polynomial-time
algorithm for deciding whether such a realization exists. Whenever the answer is affirmative, our
algorithms construct, in polynomial time, a DAG  and phylogenetic network that realizes all required
constraints in $R$ while satisfying the corresponding interpretation of the forbidden constraints
in $F$. 
The algorithms developed in this paper are implemented in the freely available Python package \texttt{RealLCA}, which is available at
\url{https://github.com/AnnaLindeberg/RealLCA}.

A central role in our analysis is played by the closure $\Fcl(R)$ of $R$ with respect to $F$. We
showed that $\Fcl(R)$ can be computed in polynomial time by repeatedly applying a small collection
of simple inference rules. This closure captures the consequences of the required and forbidden
constraints and forms the basis for both the characterization of \AF-realizability and the
construction of the associated canonical DAG. Moreover, for every \AF-realizable pair $(R,F)$, we
proved that
\[
\Fcl(R)
=
\bigcap_{G \text{ \AF-realizes } (R,F)} \rel_G .
\]
Thus, the combinatorially defined closure $\Fcl(R)$ coincides precisely with the set of 
LCA-constraints satisfied by every \AF-realization of $(R,F)$. This, in particular, shows that the
inference rules capture all LCA-constraints logically implied by $R$ and $F$.

The constructions and algorithms presented in this work are primarily intended to establish the
polynomial-time solvability of the considered realization problems. In particular, the procedures
obtained from our proofs are not claimed to be optimal with respect to their runtime. We
expect that more efficient implementations and considerably sharper runtime bounds can be obtained.
Developing such algorithms, as well as determining the precise computational complexity of the
individual construction steps, remains an interesting direction for future work.

While our results settle the general realization problems for required and forbidden 
LCA-constraints, the networks produced by our constructions are not required to belong to any
restricted network class and may therefore be structurally complex. The results developed here
nevertheless provide a foundation for a systematic investigation of restricted realization
problems. Such restrictions may require strengthening the conditions \axiom{Y1} and \axiom{Y2},
supplementing them with additional structural conditions, or replacing the canonical construction
by one tailored to the network class under consideration.

A natural next step is therefore to determine when a given pair $(R,F)$ can be realized by
specific, biologically motivated classes of phylogenetic networks. Of particular interest are
normal networks \cite{Willson:10b,Francis:25,francis2021normalising}, level-$1$ networks
\cite{Gambette2017,Hellmuth2023}, and galled trees~\cite{Gusfield:03}. For each of these classes,
one may ask for structural characterizations analogous to those obtained here, efficient
recognition and construction algorithms, and class-specific closure rules that capture precisely
the constraints shared by all admissible realizations. More generally, understanding how the
combinatorial properties of $(R,F)$ constrain the structural complexity of its realizations
provides a broad range of questions for further research.

\section*{Acknowledgments}

We thank Anna Lindeberg and Anton Alfonsson for stimulating discussions on this topic.

\section*{Declaration of Competing Interest}
The authors declare that they have no known competing financial interests or personal relationships
that could have appeared to influence the work reported in this paper.

\section*{Declaration of Generative AI in the Manuscript Preparation Process}

During the preparation of this work, the authors used ChatGPT by OpenAI to assist with literature
searches, improve grammar and spelling, streamline the text, and independently recheck selected
proofs. Following the use of this tool, the authors reviewed and revised the content as necessary
and take full responsibility for the content of the published article.

\bibliographystyle{spbasic}
\bibliography{common}

\end{document}